\documentclass[11pt]{article}
\usepackage{local_style}

\title{Fundamental Limits of Throughput and Availability: \\ \Large Applications to prophet inequalities \& transaction fee mechanism design}
\author{Aadityan Ganesh \\ Princeton University \\ aadityanganesh@princeton.edu \and Jason Hartline \\ Northwestern University \\ hartline@northwestern.edu
 \and Atanu R Sinha \\ Adobe Research \\ atr@adobe.com \and Matthew vonAllmen \\ Northwestern University \\ matthewvonallmen2026@u.northwestern.edu}\date{}

\begin{document}
\maketitle

\begin{abstract}
This paper studies the fundamental limits of availability and throughput for independent and heterogeneous demands of a limited resource. Availability is the probability that the demands are below the capacity of the resource. Throughput is the expected fraction of the resource that is utilized by the demands. We offer a concentration inequality generator that gives lower bounds on feasible availability and throughput pairs with a given capacity and independent but not necessarily identical distributions of up-to-unit demands. We show that availability and throughput cannot both be poor. These bounds are analogous to tail inequalities on sums of independent random variables, but hold throughout the support of the demand distribution. This analysis gives analytically tractable bounds supporting the unit-demand characterization of \textcite{CDL-23} and generalizes to up-to-unit demands. Our bounds also provide an approach towards improved multi-unit prophet inequalities \parencite{HKS-07}. They have applications to transaction fee mechanism design (for blockchains) where high availability limits the probability of profitable user-miner coalitions \parencite{CS-23}.
\end{abstract}


\section{Introduction}

Consider the following {\em Soup Kitchen Problem}:
\begin{quote}
    A soup kitchen produces $\threshold$ units of soup, and serves diners with independent but not necessarily identically distributed demands of at most one unit. An auditor needs to ensure that the soup kitchen is performing well. Two key performance metrics are availability $\availability$ and throughput $\throughput$. Availability is the probability the kitchen does not run out of soup. Throughput is the expected fraction of soup produced that is consumed. The auditor's problem is to catch an underperforming soup kitchen by establishing a lower bound on feasible availability and throughput pairs.
\end{quote}

We define the performant set of availability-throughput pairs $(\availability,\throughput)$ for supply $\threshold$ as a subset of $[0, 1]^{2}$ for which there exist independent distributions for the diners' demands $D_{i} \in [0, 1]$ with sum $D = \sum_{i = 1}^{n} D_{i}$ that induce at most availability and throughput of $(\availability,\throughput)$, i.e., an availability of $\availability \geq \mathbb{P}(D < \threshold)$ and a throughput of $\throughput \geq \mathbb{E}[\min(D / \threshold, 1)]$. The auditor seeks to find as tight a bound as possible on performant $(\availability, \throughput)$ pairs so that underperforming pairs are not mistaken for performant pairs. An illustration of how tighter bounds permit the auditor to identify underperforming soup kitchens is provided in \Cref{fig:toy-problem-illustration}.

\begin{figure}[h]
    \center
    \includegraphics[scale=0.65]{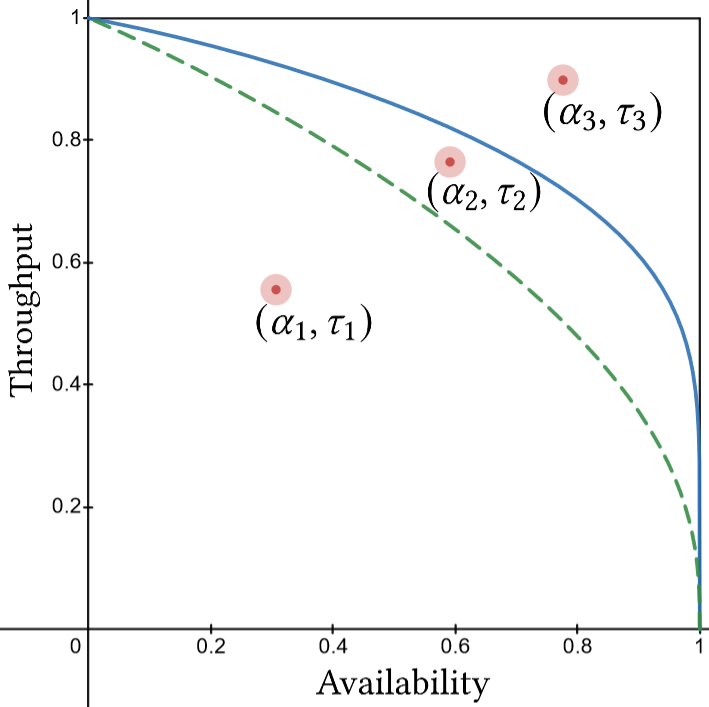}
    \caption{The solid blue line represents the true bound on performant (availability, throughput) pairs, while the dashed green line is a lower bound known to the auditor. The soup kitchen achieving $(\availability_{3}, \throughput_{3})$ above the solid blue line performs well, and the soup kitchens achieving $(\availability_{1}, \throughput_{1})$ and $(\availability_{2}, \throughput_{2})$ below the solid blue line are underperforming. However, the auditor can only detect that $(\availability_{1}, \throughput_{1})$ is underperforming, since it does not know the true location of the solid blue line.}
    \label{fig:toy-problem-illustration}
\end{figure}

Both throughput and availability are important. If a soup kitchen has low throughput, much of the soup produced each day goes to waste. However, if it has low availability, prospective diners have good reason to think that the kitchen will not have any soup to serve them, and so cannot reliably visit the kitchen and expect to be served.

This paper provides a method for generating novel concentration inequalities to relate availability and throughput. These concentration inequalities are applicable in a wide variety of domains, and improve upon similar concentration inequalities, such as the Chernoff bound, in ways that are required for a good approximate solution to the soup kitchen problem. Unlike most concentration inequalities, which can only bound the probability that a random variable falls in a region above its mean, the concentration inequalities provided in this paper can bound the probability that a random variable falls in any upward-closed region. In this sense, the concentration inequalities we provide move beyond "tail bounds" and can instead bound a larger class of regions than the tail of a distribution.

A mechanism designer that can fine-tune either the availability or the throughput can use these bounds to establish worst-case values for whichever parameter they do not fine-tune. When demands come from price-sensitive agents, the designer can increase the price to decrease throughput and increase availability, while decreasing the price does the opposite. Furthermore, since the bounds improve as the size of the soup supply $\threshold$ increases, a designer that increases their supply while targeting a specific availability can increase the value of their worst-case throughput. Alternatively, a designer that increases their supply while targeting a specific throughput can increase the value of their worst-case availability. This allows them to identify the minimal amount of soup $\threshold$ which permits them to meet their desired values for availability and throughput.

The fact that posted prices can be used to pick a particular point on the tradeoff curve between availability and throughput is particularly useful for prophet inequalities. The conventional way to prove multi-unit prophet inequalities is to demonstrate the existence of a price that achieves a point on the the tradeoff curve which maximizes the posted-price mechanism's worst-case expected welfare. Better bounds on the worst-case throughput in terms of the availability, or better bounds on the worst-case availability in terms of the throughput, translate into better bounds on expected welfare.

\subsection*{Results}
Our primary theorem is a concentration inequality generator. Given a function $f$ that satisfies certain criteria, it constructs a concentration inequality that is either stronger or weaker (and more or less tractable) depending on the function $f$ selected.

\newcommand{\maintheoremtext}{Let supply $\threshold \geq 0$, and let $D_{1}, ..., D_{n} \in [0, 1]$ be independent random demands whose sum $D = \sum_{i = 1}^{n} D_{i}$ has throughput $\throughput = \mathbb{E}[\min(D / \threshold, 1)]$ and availability $\availability = \mathbb{P}(D < \threshold)$. For any weakly convex and weakly positive function $f: \mathbb{R}_{+} \rightarrow \mathbb{R}$ that is strictly increasing at $\threshold$, the probability that $D$ is at least $\threshold$ is bounded above by the expectation of $f$ evaluated on a binomially distributed random variable with mean equal to the absolute throughput $\threshold \throughput$, normalized by $f$ evaluated at the supply $\threshold$. I.e., for $X \sim \text{Binomial}(n, \threshold \throughput / n)$,
    \begin{align*}
        1 - \availability \leq \frac{\mathbb{E}[f(X)]}{f(\threshold)} \text{.}
    \end{align*}

}

\begin{customthm}{\ref{main_theorem}}
    \maintheoremtext
\end{customthm}

Note that, unlike Chernoff bounds, the mean of the demand distribution $\mathbb{E}[D]$ does not appear in concentration inequalities generated by \Cref{main_theorem}. The only quantities relating to the demand $D$ which appear are the availability and the throughput. Furthermore, observe that \Cref{main_theorem} permits the selection of thresholds $\threshold$ which are less than the mean of the demand $D$, whereas Chernoff bounds do not. Thresholds $\threshold$ below $\mathbb{E}[D]$ permit concentration inequalities generated by \Cref{main_theorem} to bound the probability of non-tail regions of the demand $D$.

To prove \Cref{main_theorem}, we first establish that the expected value of a convex non-negative function of demand is non-decreasing as we go from general up-to-unit demand to asymmetric unit-demand to binomial (symmetric unit-demand) to Poisson demand (an infinite number of unit demands)\footnote{This last inequality between binomial demand and Poisson demand is eventually used to show that the formula in \Cref{main_theorem} is weakest when the number of demands $n \to \infty$.}. Next we show that the probability of the demand exceeding the threshold $\threshold$ can be bounded by the expected value of a non-negative convex function $f$ of a binomial distribution with the same expectation as $D$, normalized by the value of the function on the expected overdemand, i.e. $f(\mathbb{E}[D \mid D \geq \threshold])$. Finally, we show that this bound still holds when all instances of the expected overdemand $\mathbb{E}[D \mid D \geq \threshold]$ are lowered to the threshold $\kappa$, producing the formula for \Cref{main_theorem}.

A simple argument then shows that the tightest bound for \Cref{main_theorem} comes from a weakly convex function that is zero up to a point and then linear, i.e., $f(x) = \max(x - \reluparam, 0)$ for some constant $\reluparam$. In machine learning this is referred to as a ReLU function. It is the case that non-ReLU bounds, though less tight than ReLU bounds, may still be desirable due to their comparative tractability. We prove several bounds, each produced from a different function $f$, and compare them to the optimal ReLU bound, showing how different bounds trade off tractability against strength.

Lastly, we apply one of these bounds to the setting of multi-unit prophet inequalities with up-to-unit demands. Conventionally, multi-unit prophet inequalities have attempted to showcase the existence of posted prices that produce good expected welfare relative to the optimal allocation, which requires selecting posted prices that sacrifice availability for throughput. The literature on multi-unit prophet inequalities has not explored what occurs when availability is its own independent concern for the mechanism designer, and we show in \Cref{invertible_near_chernoff_bound} that one of the bounds generated by \Cref{main_theorem} can be used to illustrate a tractable yet strong bound on the tradeoff between stronger prophet inequalities and availability.

\subsection*{Characterization of Throughput versus Availability}

A consequence of the analysis of multi-unit prophet inequalities in \textcite{CDL-23} is a simple (but analytically intractable) characterization of the limits of throughput and availability for unit demands, i.e., when each individual demand $D_{i}$ is in $\{ 0, 1 \}$, showing that for any given availability $\availability$ the worst-case/smallest throughput $\throughput$ occurs when the total demand $D$ is Poisson-distributed. Therefore, there does not exist any feasible $(\availability, \throughput)$ pairs below the curve drawn out by assuming $D$ is Poisson distributed and varying the Poisson parameter $\mathbb{E}[D]$ from $\mathbb{E}[D] = 0$, where $(\availability, \throughput) = (1, 0)$, to $\mathbb{E}[D] \to \infty$, where $(\availability, \throughput) = (0, 1)$. Poisson distributions for $D$ with a mean between these two extremes will trade off between different intermediate values for throughput and availability.

As mentioned, exact analysis of Poisson tail probabilities are not analytically tractable; that is, given any particular availability, there is no closed form for the corresponding throughput of a Poisson distribution. Our analysis provides concentration inequalities that do give a closed form for worst-case throughput in terms of availability. In addition, our bounds are not restricted to the unit demand setting of multi-unit prophet inequalities. They also hold in a more general setting with up-to-unit demands, i.e. where each $D_{i}$ can take on values in the unit interval $[0, 1]$. In the setting with up-to-unit demand diners, a simple characterization of the solution to the soup kitchen problem is not known. We conjecture that the worst-case distribution in the up-to-unit-demand setting is also Poisson:

\begin{conjecture}
    For independently distributed up-to-unit demands with availability and throughput $(\availability, \throughput)$ for a given supply, there is a Poisson distribution for total demand with availability and throughput at most $(\availability, \throughput)$.
\end{conjecture}

Throughout this paper we will use availability-throughput pairs corresponding to Poisson distributions as a benchmark to compare against our concentration inequalities.

\subsection*{Applications}
The analysis presented of the soup kitchen problem is applicable in several economic and computational settings.

Consider airline ticket sales, where a clear tradeoff exists between throughput and availability \parencite{bertsimas2005simulation,boyd2016future}. Airlines intentionally overbook tickets for seats on a flight knowing that some customers will cancel or not show up~\parencite{suzuki2002empirical}. When more passengers show up than there are seats, an airline must engage in "bumping," whereby some passengers are removed from the flight, either by way of offering them a voucher for a new flight or rescheduling them for a new one later that day. An airline that tried to set its availability $\availability = 1$ will never overbook, but will waste most of the seats on a flight by frequently leaving them empty, exemplified by a low throughput value. Conversely, an airline that aggressively overbooks will achieve a high throughput, but will need to bump passengers so frequently that airline tickets are no longer seen as a reliable indicator for a passenger that they will be boarding a flight at the time printed upon their ticket. Balancing these two quantities is an important tradeoff for the airline to make, and they can control where they fall along the tradeoff curve by changing the degree of overbooking in which they engage.

Analogous tradeoffs occur commonly in cloud computing~\parencite{lee2011heterogeneity,lu2016fairness} with high fixed cost of compute and storage, and heterogeneity in demand from users. An exemplar is that of major enterprises renting costly GPUs from cloud providers to provision dedicated GPUs to internal users or products~\parencite{jeon2018multi}. Consider a contract where the enterprise rents a fixed supply of GPUs from a provider, and then in turn provisions them out to its employees on demand. Employees' requests for GPUs arrive in sequence and different requests have different values in terms of productivity for the enterprise. The enterprise would like to maximize its throughput, since it still has to pay for each of the GPUs regardless of whether an employee makes use of them. At the same time, the enterprise also desires high availability. If employees request GPUs and the enterprise consistently cannot fulfil those requests because it has exhausted its supply, it is unproductive for the enterprise. An enterprise can control where it falls on the availability-throughput tradeoff curve by making GPUs easier or harder for employees to acquire. Additionally, rather than trading off between availability and throughput, the enterprise can potentially increase both metrics by increasing the supply of GPUs $\threshold$, which improves the entire availability-throughput tradeoff curve. Since it is costly to rent GPUs from a cloud provider, an enterprise may wish to select a minimal supply $\threshold$ that still permits it to meet its desired availability and throughput values.

Another application arises from designing transaction fee mechanisms for blockchains. A block can accommodate up to $\threshold$ units of data. The block proposer packs the block with transactions performed by users. Different transactions require different quantities of data \parencite{VB-14}, and typically the total demand is more than the space available on the block. In the event this occurs, a transaction fee mechanism is deployed to determine which transactions are included in the block. A central question explored in the literature is to design a mechanism that satisfies strategy-proofness for users and block proposers, while also deterring collusion between block proposers and users \parencite{TR-20, CS-23}. Ethereum, for instance, runs a posted-price mechanism that burns all the payments collected from users \parencite{TR-20, TR-21}. Whenever the demand at the base price set by the protocol is less than the capacity of the block, the mechanism satisfies all the required properties. On the other hand, when demand exceeds the block capacity, space in the block is allocated through an emergency mechanism, which is a first-price auction in Ethereum's case. First-price auctions are not strategy-proof for users, while switching to a second-price auction introduces strategic deviations for block proposers. \textcite{CS-23} show, under some natural assumptions, that there cannot exist a mechanism that guarantees strategy-proofness and collusion resistance for all parties simultaneously. One remedy would be to instead guarantee strategy-proofness and collusion resistance with high probability \parencite[cf.][]{GH-03}. High availability corresponds to a lower probability of these periods of over-demand, thereby mitigating the need to run the "undesirable" emergency auction, while a larger throughput corresponds to efficient space utilization in the block (or a lower quantity of information for block proposers to process in the worst-case). The base price can be adjusted to trade off availability (and the strategy-proofness that comes with it) and throughput (and the efficiency that comes with it).

\subsection*{Related Work}

\textcite{HKS-07} apply a prophet inequality to posted-price mechanisms that ration a supply of $\threshold$ units of an item among $n$ agents with independent but not identically distributed values for those items. They showcase an upper bound of $1 + \sqrt{\nicefrac{8 \ln(\threshold)}{\threshold}}$ and a lower bound of $1 + \sqrt{\nicefrac{1}{512 \threshold}}$ on the approximation ratio for expected welfare achieved by a posted-price mechanism compared to a prophet's optimal allocation. Their setting assumes binary demand; either agents purchase a single unit of the item, or do not purchase at all. Additionally, the upper and lower bounds derived are only valid for high values of $\threshold$. \textcite{JJM-23} confirm the asymptotic optimality of the former bound, showing that the fraction of the expected welfare achieved by a posted-price mechanism compared to a prophet is $1 - \Theta(\sqrt{\log(\threshold) / \threshold})$.

\textcite{CDL-23} analyze the same setting and conclude that the worst-case distribution of demand is Poisson, where there are infinitely many agents who all purchase with the same infinitesimally small probability. This result enables numerical calculation of the worst case approximation factor, specifically in the small $\threshold$ case. Though it is not the main focus of their paper, their analysis that Poisson demand is the worst case for welfare also demonstrates, for our paper, that Poisson demand produces the worst case throughput for any given availability when individuals have unit demands.

Our main theoretical contribution, the development of tail bounds superior to the Chernoff bound, also applies to the special case of individuals with unit demands. This is useful because the tradeoff curve between availability and throughput of the worst-case Poisson demand distribution is not analytically tractable. In contrast, several of the bounds we develop have a closed form.

\section{The Soup Kitchen Problem}

\subsection{Model}
The soup kitchen produces $\threshold$ units of soup. The kitchen serves $n$ diners with independent but not necessarily identically distributed demands. Let diner $i$'s demand be $D_{i} \in [0, 1]$ and $D = \sum_{i = 1}^{n} D_i$ be the total demand. The availability $\availability = \mathbb{P}(D < \threshold)$ of the kitchen is the probability that the total demand can be served without running out of soup. Throughput $\throughput = \mathbb{E}[\min(D / \threshold, 1)]$ is the expected fraction of the produced soup that is consumed. Additionally, we will sometimes refer to the quantity $\threshold \throughput = \mathbb{E}[\min(D, \threshold)]$ as the absolute throughput, i.e., the expected total units of soup consumed. The goal is to generate concentration bounds on pairs of availability and throughput $(\availability, \throughput)$ that preclude the possibility of any pair below the bound being realized by independent diner demands.

\subsection{Limits of Throughput and Availability}

The following high-level sequence of results culminates in our main theorem.

\begin{theoremrep}\label{numerator_inequality}
    Let $D_{1}, ..., D_{n} \in [0, 1]$ be independent random demands with sum $D = \sum_{i = 1}^{n} D_{i}$ and collective mean $\mu = \mathbb{E}[D]$. Then, for any weakly convex function $f: \mathbb{R}_{+} \rightarrow \mathbb{R}$, the expectation of $f(D)$ is maximized when $D$ is binomially distributed, and is further bounded above when $D$ is Poisson distributed with the same mean. I.e., for $X \sim \text{Binomial}(n, \mu / n)$ and $Y \sim \text{Poisson}(\mu)$,
    \begin{align*}
        \mathbb{E}[f(D)] \leq \mathbb{E}[f(X)] \leq \mathbb{E}[f(Y)] \text{.}
    \end{align*}
\end{theoremrep}

\begin{proof}
    Proving \Cref{numerator_inequality} can be split into the above two separate steps. The first of these two steps is \Cref{sum_of_bernoullis}, proved subsequently. Consider that when we select $L = 0$, \Cref{sum_of_bernoullis} tells us that
    \begin{align*}
        \mathbb{E}[f(D)] \leq \sum_{S \subseteq [n]} f(|S|) \left[ \prod_{i \in S} \mu_{i} \right] \left[ \prod_{i \in [n] \setminus S} \left( 1 - \mu_{i} \right) \right] \text{.}
    \end{align*}
    for $\mu_{1}, ..., \mu_{n}$ denoting the means of the individual random variables $D_{1}, ..., D_{n}$.
    
    Observe that \Cref{iid_is_worst} also allows us to show
    \begin{align*}
        \sum_{S \subseteq [n]} f(|S|) \left[ \prod_{i \in S} \frac{\mu}{n} \right] \left[ \prod_{i \in [n] \setminus S} \left( 1 - \frac{\mu}{n} \right) \right] \leq \sum_{i = 0}^{\infty} f(i) \exp\left( -\mu \right) \frac{\mu^{i}}{i!} \text{.}
    \end{align*}
    To see this, let $n' \gg n$ and begin with $\mu_{i} = \mu / n$ for all $i \in [n]$ and $\mu_{i} = 0$ for all $i \in [n'] \setminus [n]$. Then, the procedure laid out in \Cref{iid_is_worst} lays out a path by which all $\mu_{i}$ converge to $\mu / n'$ and demonstrates
    \begin{align*}
        \sum_{S \subseteq [n]} f(|S|) \left[ \prod_{i \in S} \frac{\mu}{n} \right] \left[ \prod_{i \in [n] \setminus S} \left( 1 - \frac{\mu}{n} \right) \right] \leq \sum_{S \subseteq [n']} f(|S|) \left[ \prod_{i \in S} \frac{\mu}{n'} \right] \left[ \prod_{i \in [n'] \setminus S} \left( 1 - \frac{\mu}{n'} \right) \right] \text{.}
    \end{align*}
    Furthermore, observe that by the definition of a binomial distribution,
    \begin{align*}
        &\mathbin{\phantom{=}} \sum_{S \subseteq [n]} f(|S|) \left[ \prod_{i \in S} \frac{\mu}{n} \right] \left[ \prod_{i \in [n] \setminus S} \left( 1 - \frac{\mu}{n} \right) \right]
        \\ &= \sum_{i = 0}^{n} f(i) \frac{n!}{i! (n - i)!} \left( \frac{\mu}{n} \right)^{i} \left( 1 - \frac{\mu}{n} \right)^{n - i} = \mathbb{E}[f(X)]
    \end{align*}
    for $X \sim \text{Binomial}(n, \mu / n)$. Therefore, as we've demonstrated that $\mathbb{E}[f(X)]$ is weakly increasing in $n$, we have $\mathbb{E}[f(X)] \leq \mathbb{E}[f(Y)]$ for $Y \sim \text{Poisson}(\mu)$, which is the limiting case as $n \to \infty$. We conclude that \Cref{numerator_inequality} must hold, i.e. for our initial random variable $D$ that is the sum of $n$ independent $D_{i} \in [0, 1]$,
    \begin{align*}
        \mathbb{E}[f(D)] \leq \mathbb{E}[f(X)] \leq \mathbb{E}[f(Y)] \text{.}
    \end{align*}
\end{proof}

\begin{theorem}\label{premarkov_tail_bound}
    Let supply $\threshold \geq 0$, and let $D_{1}, ..., D_{n} \in [0, 1]$ be independent random demands whose sum $D = \sum_{i = 1}^{n} D_{i}$ has collective mean $\mu = \mathbb{E}[D]$ and availability $\availability = \mathbb{P}(D < \threshold)$. Then, for any weakly convex and weakly positive function $f: \mathbb{R}_{+} \rightarrow \mathbb{R}$, the probability $D$ is at least $\threshold$ is bounded above by the expectation of $f$ evaluated on a binomially distributed random variable with the same mean as $D$, normalized by $f$ evaluated at the conditional mean of $D$ given that it is at least as large as $\threshold$. I.e., for $X \sim \text{Binomial}(n, \mu / n)$,
    \begin{align*}
        1 - \availability \leq \frac{\mathbb{E}[f(X)]}{f(\mathbb{E}[D \mid D \geq \threshold])} \text{.}
    \end{align*}
\end{theorem}

\begin{theoremrep}\label{main_theorem}
    \maintheoremtext
\end{theoremrep}

\begin{proof}
    We first introduce some notation. Let $\mu \triangleq \mathbb{E}[D]$, and denote the upper conditional mean $\mu^{\uparrow} \triangleq \mathbb{E}[D \mid D \geq \threshold]$ and the quantity $\mu^{\downarrow} \triangleq \mathbb{E}[D \mathds{1}(D < \threshold)]$; thus, $\mu = \mu^{\uparrow} \mathbb{P}(D \geq \threshold) + \mu^{\downarrow}$.

    We also define the subderivative $f'\left( \threshold^{+} \right) \equiv \lim_{z \downarrow \threshold} \frac{f(z) - f(\threshold)}{z - \threshold}$, which we assume to be strictly positive.

    Now, consider the mapping $x \mapsto \max(x - \reluparam, 0)$ for $\reluparam = \threshold - \frac{f(\threshold)}{f'(\threshold^{+})}$, which satisfies the requirements of \Cref{premarkov_tail_bound}. Then, explicitly writing out the expectation in the inequality provided by \Cref{premarkov_tail_bound} tells us that
    \begin{align*}
        \mathbb{P}(D \geq \threshold) &\leq \frac{\sum_{i = 0}^{n} \max(i - \reluparam, 0) \frac{n!}{i! (n - i)!} \left( \frac{\mu}{n} \right)^{i} \left( 1 - \frac{\mu}{n} \right)^{n - i}}{\mu^{\uparrow} - \reluparam} \text{,}
    \end{align*}
    where we write the denominator as $\mu^{\uparrow} - \reluparam$ rather than $\max(\mu^{\uparrow} - \reluparam, 0)$ because $\mu^{\uparrow} - \reluparam \geq 0$. We can then convert the right-hand side of the above inequality into a function of arbitrary parameters $u^{\uparrow}$ and $u^{\downarrow}$ that replace $\mu^{\uparrow}$ and $\mu^{\downarrow}$, respectively. To do this, we define the function
    \begin{align*}
        G(u^{\uparrow}, u^{\downarrow}) = \frac{\sum_{i = 0}^{n} \max(i - \reluparam, 0) \frac{n!}{i! (n - i)!} \left( \frac{u^{\uparrow} \mathbb{P}(D \geq \threshold) + u^{\downarrow}}{n} \right)^{i} \left( 1 - \frac{u^{\uparrow} \mathbb{P}(D \geq \threshold) + u^{\downarrow}}{n} \right)^{n - i}}{u^{\uparrow} - \reluparam}
    \end{align*}
    and observe that $\mathbb{P}(D \geq \threshold) \leq G(\mu^{\uparrow}, \mu^{\downarrow})$.

    We will manipulate the parameters of this function $G$ to show that $G(\mu^{\uparrow}, \mu^{\downarrow}) \leq G(\threshold, \mu^{\downarrow})$, allowing us to conclude that $\mathbb{P}(D \geq \threshold) \leq G(\threshold, \mu^{\downarrow})$. To do this, we will examine two cases: where $\reluparam \geq 0$, and where $\reluparam < 0$. In the case where $\reluparam \geq 0$, we will first showcase a continuous path in $(u^{\uparrow}, u^{\downarrow})$ space that preserves the value of $G$ and goes from $(\mu^{\uparrow}, \mu^{\downarrow})$ to $(\threshold, t)$ for some $t \leq \mu^{\downarrow}$. Then, we will apply the fact that $G(\threshold, \cdot)$ is weakly increasing to show $G(\mu^{\uparrow}, \mu^{\downarrow}) \leq G(\threshold, \mu^{\downarrow})$. In the case where $\reluparam < 0$, we will instead show directly that $G(\cdot, \mu^{\downarrow})$ is decreasing.

    Assume that $\reluparam \geq 0$. Observe that for any fixed $u^{\uparrow} \in (\reluparam, \mu^{\uparrow}]$, the function $G(u^{\uparrow}, \cdot)$ is continuous across all inputs $u^{\downarrow}$ such that
    \begin{align*}
        u^{\uparrow} \mathbb{P}(D \geq \threshold) + u^{\downarrow} \in [0, n] \text{.}
    \end{align*}
    In other words, $G(u^{\uparrow}, \cdot)$ is continuous in the closed interval $[-u^{\uparrow} \mathbb{P}(D \geq \threshold), n - u^{\uparrow} \mathbb{P}(D \geq \threshold)]$.
    
    Note that the value $u^{\downarrow} = (\mu^{\uparrow} - u^{\uparrow}) \mathbb{P}(D \geq \threshold) + \mu^{\downarrow}$ lies in this interval, which we briefly verify by checking that
    \begin{align*}
        u^{\uparrow} \mathbb{P}(D \geq \threshold) + (\mu^{\uparrow} - u^{\uparrow}) \mathbb{P}(D \geq \threshold) + \mu^{\downarrow}
        &= \mu^{\uparrow} \mathbb{P}(D \geq \threshold) + \mu^{\downarrow}
        \\ &= \mu \triangleq \mathbb{E}[D] \in [0, n] \text{.}
    \end{align*}
    It follows that since $G(u^{\uparrow}, \cdot)$ is continuous in the interval $[-u^{\uparrow} \mathbb{P}(D \geq \threshold), n - u^{\uparrow} \mathbb{P}(D \geq \threshold)]$, it is also continuous in the subinterval $[-u^{\uparrow} \mathbb{P}(D \geq \threshold), (\mu^{\uparrow} - u^{\uparrow}) \mathbb{P}(D \geq \threshold) + \mu^{\downarrow}]$.
    
    We now observe that when $u^{\downarrow}$ is equal to the endpoints of this interval, $-u^{\uparrow} \mathbb{P}(D \geq \threshold)$ and $(\mu^{\uparrow} - u^{\uparrow}) \mathbb{P}(D \geq \threshold) + \mu^{\downarrow}$, we obtain values of $G$ that are lower and higher than $G(\mu^{\uparrow}, \mu^{\downarrow})$, respectively. First, when $u^{\downarrow} = -u^{\uparrow} \mathbb{P}(D \geq \threshold)$, we have
    \begin{align*}
        G(u^{\uparrow}, -u^{\uparrow} \mathbb{P}(D \geq \threshold)) &= \frac{\sum_{i = 0}^{n} \max(i - \reluparam, 0) \frac{n!}{i! (n - i)!} \left( \frac{0}{n} \right)^{i} \left( 1 - \frac{0}{n} \right)^{n - i}}{u^{\uparrow} - \reluparam}
        \\ &= \frac{\max(0 - \reluparam, 0)}{u^{\uparrow} - \reluparam}
        \\ &= 0
        \\ &\leq G(\mu^{\uparrow}, \mu^{\downarrow})
    \end{align*}
    and
    \begin{align*}
        G(u^{\uparrow}, (\mu^{\uparrow} - u^{\uparrow}) \mathbb{P}(D \geq \threshold) + \mu^{\downarrow}) &= \frac{\sum_{i = 0}^{n} \max(i - \reluparam, 0) \frac{n!}{i! (n - i)!} \left( \frac{\mu}{n} \right)^{i} \left( 1 - \frac{\mu}{n} \right)^{n - i}}{u^{\uparrow} - \reluparam}
        \\ &\geq \frac{\sum_{i = 0}^{n} \max(i - \reluparam, 0) \frac{n!}{i! (n - i)!} \left( \frac{\mu}{n} \right)^{i} \left( 1 - \frac{\mu}{n} \right)^{n - i}}{\mu^{\uparrow} - \reluparam}
        \\ &= G(\mu^{\uparrow}, \mu^{\downarrow}) \text{.}
    \end{align*}

    Therefore, by the intermediate value theorem, there will always exist some $u^{\downarrow}$ within the interval $[-u^{\uparrow} \mathbb{P}(D \geq \threshold), (\mu^{\uparrow} - u^{\uparrow}) \mathbb{P}(D \geq \threshold) + \mu^{\downarrow}]$ such that $G(u^{\uparrow}, u^{\downarrow})$ exactly equals $G(\mu^{\uparrow}, \mu^{\downarrow})$. We will hereafter write $u^{\downarrow}$ as a function of $u^{\uparrow}$ so as to always select a $u^{\downarrow}$ where this equality holds:
    \begin{align*}
        G(u^{\uparrow}, u^{\downarrow}(u^{\uparrow})) = G(\mu^{\uparrow}, \mu^{\downarrow})
    \end{align*}
    By definition, varying $u^{\uparrow}$ keeps the left-hand side of the above equality constant. Thus, if we take the total derivative of $G$ with respect to $u^{\uparrow}$, we will always have
    \begin{align*}
        \frac{dG(u^{\uparrow}, u^{\downarrow}(u^{\uparrow}))}{du^{\uparrow}} = 0 \text{.}
    \end{align*}
    We will use this fact to demonstrate that $u^{\downarrow}(u^{\uparrow})$ is weakly increasing in $u^{\uparrow}$. To do so, we compute the total derivative as follows:
    \begin{align*}
        &\mathbin{\phantom{=}} \frac{dG(u^{\uparrow}, u^{\downarrow}(u^{\uparrow}))}{du^{\uparrow}}
        \\ &= \frac{\partial G}{\partial u^{\uparrow}}(u^{\uparrow}, u^{\downarrow}(u^{\uparrow})) + \frac{\partial G}{\partial u^{\downarrow}}(u^{\uparrow}, u^{\downarrow}(u^{\uparrow})) \frac{du^{\downarrow}(u^{\uparrow})}{du^{\uparrow}}
        \\ &= \frac{\sum_{i = 0}^{n} \max(i - \reluparam, 0) \frac{n!}{i! (n - i)!} \frac{\partial \left[ \left( \frac{u^{\uparrow} \mathbb{P}(D \geq \kappa) + u^{\downarrow}(u^{\uparrow})}{n} \right)^{i} \left( 1 - \frac{u^{\uparrow} \mathbb{P}(D \geq \kappa) + u^{\downarrow}(u^{\uparrow})}{n} \right)^{n - i} \right]}{\partial u^{\uparrow}}}{u^{\uparrow} - \reluparam}
        \\ &\mathbin{\phantom{=}} - \frac{\sum_{i = 0}^{n} \max(i - \reluparam, 0) \frac{n!}{i! (n - i)!} \left( \frac{u^{\uparrow} \mathbb{P}(D \geq \kappa) + u^{\downarrow}(u^{\uparrow})}{n} \right)^{i} \left( 1 - \frac{u^{\uparrow} \mathbb{P}(D \geq \kappa) + u^{\downarrow}(u^{\uparrow})}{n} \right)^{n - i}}{(u^{\uparrow} - \reluparam)^{2}}
        \\ &\mathbin{\phantom{=}} + \frac{\partial G}{\partial u^{\downarrow}}(u^{\uparrow}, u^{\downarrow}(u^{\uparrow})) \frac{du^{\downarrow}(u^{\uparrow})}{du^{\uparrow}}
    \end{align*}
    and since $\frac{dG(u^{\uparrow}, u^{\downarrow}(u^{\uparrow}))}{du^{\uparrow}} = 0$, we can rearrange the above equality into the form
    \begin{align*}
        &\mathbin{\phantom{=}} \frac{\partial G}{\partial u^{\downarrow}}(u^{\uparrow}, u^{\downarrow}(u^{\uparrow})) \frac{du^{\downarrow}(u^{\uparrow})}{du^{\uparrow}}
        \\ &= \frac{\sum_{i = 0}^{n} \max(i - \reluparam, 0) \frac{n!}{i! (n - i)!} \left( \frac{u^{\uparrow} \mathbb{P}(D \geq \kappa) + u^{\downarrow}(u^{\uparrow})}{n} \right)^{i} \left( 1 - \frac{u^{\uparrow} \mathbb{P}(D \geq \kappa) + u^{\downarrow}(u^{\uparrow})}{n} \right)^{n - i}}{(u^{\uparrow} - \reluparam)^{2}}
        \\ &\mathbin{\phantom{=}} - \frac{\sum_{i = 0}^{n} \max(i - \reluparam, 0) \frac{n!}{i! (n - i)!} \frac{\partial \left[ \left( \frac{u^{\uparrow} \mathbb{P}(D \geq \kappa) + u^{\downarrow}(u^{\uparrow})}{n} \right)^{i} \left( 1 - \frac{u^{\uparrow} \mathbb{P}(D \geq \kappa) + u^{\downarrow}(u^{\uparrow})}{n} \right)^{n - i} \right]}{\partial u^{\uparrow}}}{u^{\uparrow} - \reluparam} \text{.}
    \end{align*}
    Observe that the two terms on the right-hand side can be simplified, respectively, as
    \begin{align*}
        &\mathbin{\phantom{=}} \frac{\sum_{i = 0}^{n} \max(i - \reluparam, 0) \frac{n!}{i! (n - i)!} \left( \frac{u^{\uparrow} \mathbb{P}(D \geq \kappa) + u^{\downarrow}(u^{\uparrow})}{n} \right)^{i} \left( 1 - \frac{u^{\uparrow} \mathbb{P}(D \geq \kappa) + u^{\downarrow}(u^{\uparrow})}{n} \right)^{n - i}}{(u^{\uparrow} - \reluparam)^{2}}
        \\ &= \frac{1}{u^{\uparrow} - \reluparam} G(u^{\uparrow}, u^{\downarrow}(u^{\uparrow}))
        \\ &= \frac{G(\mu^{\uparrow}, \mu^{\downarrow})}{u^{\uparrow} - \reluparam}
    \end{align*}
    and
    \begin{align*}
        &\mathbin{\phantom{=}} - \frac{\sum_{i = 0}^{n} \max(i - \reluparam, 0) \frac{n!}{i! (n - i)!} \frac{\partial \left[ \left( \frac{u^{\uparrow} \mathbb{P}(D \geq \kappa) + u^{\downarrow}(u^{\uparrow})}{n} \right)^{i} \left( 1 - \frac{u^{\uparrow} \mathbb{P}(D \geq \kappa) + u^{\downarrow}(u^{\uparrow})}{n} \right)^{n - i} \right]}{\partial u^{\uparrow}}}{u^{\uparrow} - \reluparam}
        \\ &= - \frac{\sum_{i = 0}^{n} \max(i - \reluparam, 0) \frac{n!}{i! (n - i)!} \frac{\partial \left[ \left( \frac{u^{\uparrow} \mathbb{P}(D \geq \kappa) + u^{\downarrow}(u^{\uparrow})}{n} \right)^{i} \right]}{\partial u^{\uparrow}} \left( 1 - \frac{u^{\uparrow} \mathbb{P}(D \geq \kappa) + u^{\downarrow}(u^{\uparrow})}{n} \right)^{n - i}}{u^{\uparrow} - \reluparam}
        \\ &\mathbin{\phantom{=}} - \frac{\sum_{i = 0}^{n} \max(i - \reluparam, 0) \frac{n!}{i! (n - i)!} \left( \frac{u^{\uparrow} \mathbb{P}(D \geq \kappa) + u^{\downarrow}(u^{\uparrow})}{n} \right)^{i} \frac{\partial \left[ \left( 1 - \frac{u^{\uparrow} \mathbb{P}(D \geq \kappa) + u^{\downarrow}(u^{\uparrow})}{n} \right)^{n - i} \right]}{\partial u^{\uparrow}}}{u^{\uparrow} - \reluparam}
        \\ &= - \frac{\sum_{i = 0}^{n} \max(i - \reluparam, 0) \frac{n!}{i! (n - i)!} (\frac{i \mathbb{P}(D \geq \kappa)}{n}) \left( \frac{u^{\uparrow} \mathbb{P}(D \geq \kappa) + u^{\downarrow}(u^{\uparrow})}{n} \right)^{i - 1} \left( 1 - \frac{u^{\uparrow} \mathbb{P}(D \geq \kappa) + u^{\downarrow}(u^{\uparrow})}{n} \right)^{n - i}}{u^{\uparrow} - \reluparam}
        \\ &\mathbin{\phantom{=}} - \frac{\sum_{i = 0}^{n} \max(i - \reluparam, 0) \frac{n!}{i! (n - i)!} (\frac{-(n - i) \mathbb{P}(D \geq \kappa)}{n}) \left( \frac{u^{\uparrow} \mathbb{P}(D \geq \kappa) + u^{\downarrow}(u^{\uparrow})}{n} \right)^{i} \left( 1 - \frac{u^{\uparrow} \mathbb{P}(D \geq \kappa) + u^{\downarrow}(u^{\uparrow})}{n} \right)^{n - i - 1}}{u^{\uparrow} - \reluparam}
        \\ &= - \mathbb{P}(D \geq \kappa) \frac{\sum_{i = 1}^{n} \max(i - \reluparam, 0) \frac{(n - 1)!}{(i - 1)! (n - i)!} \left( \frac{u^{\uparrow} \mathbb{P}(D \geq \kappa) + u^{\downarrow}(u^{\uparrow})}{n} \right)^{i - 1} \left( 1 - \frac{u^{\uparrow} \mathbb{P}(D \geq \kappa) + u^{\downarrow}(u^{\uparrow})}{n} \right)^{n - i}}{u^{\uparrow} - \reluparam}
        \\ &\mathbin{\phantom{=}} - \mathbb{P}(D \geq \kappa) \frac{\sum_{i = 0}^{n - 1} -\max(i - \reluparam, 0) \frac{(n - 1)!}{i! (n - i - 1)!} \left( \frac{u^{\uparrow} \mathbb{P}(D \geq \kappa) + u^{\downarrow}(u^{\uparrow})}{n} \right)^{i} \left( 1 - \frac{u^{\uparrow} \mathbb{P}(D \geq \kappa) + u^{\downarrow}(u^{\uparrow})}{n} \right)^{n - i - 1}}{u^{\uparrow} - \reluparam}
        \\ &= - \frac{\mathbb{P}(D \geq \kappa)}{u^{\uparrow} - \reluparam} \sum_{i = 0}^{n - 1} \left( \max(i + 1 - \reluparam, 0) - \max(i - \reluparam, 0) \right)
        \\ & \cdot \frac{(n - 1)!}{i! (n - 1 - i)!} \left( \frac{u^{\uparrow} \mathbb{P}(D \geq \kappa) + u^{\downarrow}(u^{\uparrow})}{n} \right)^{i} \left( 1 - \frac{u^{\uparrow} \mathbb{P}(D \geq \kappa) + u^{\downarrow}(u^{\uparrow})}{n} \right)^{n - 1 - i}
        \\ &\geq - \frac{\mathbb{P}(D \geq \kappa)}{u^{\uparrow} - \reluparam} \sum_{i = 0}^{n - 1} \frac{(n - 1)!}{i! (n - 1 - i)!} \left( \frac{u^{\uparrow} \mathbb{P}(D \geq \kappa) + u^{\downarrow}(u^{\uparrow})}{n} \right)^{i} \left( 1 - \frac{u^{\uparrow} \mathbb{P}(D \geq \kappa) + u^{\downarrow}(u^{\uparrow})}{n} \right)^{n - 1 - i}
        \\ &= - \frac{\mathbb{P}(D \geq \kappa)}{u^{\uparrow} - \reluparam} \text{.}
    \end{align*}
    Thus, we have
    \begin{align*}
        \frac{\partial G}{\partial u^{\downarrow}}(u^{\uparrow}, u^{\downarrow}(u^{\uparrow})) \frac{du^{\downarrow}(u^{\uparrow})}{du^{\uparrow}} \geq \frac{G(\mu^{\uparrow}, \mu^{\downarrow})}{u^{\uparrow} - \reluparam} - \frac{\mathbb{P}(D \geq \kappa)}{u^{\uparrow} - \reluparam} \geq 0 \text{.}
    \end{align*}
    Lastly, note that $G(u^{\uparrow}, \cdot)$ is a weakly increasing function, as the numerator of
    \begin{align*}
        G(u^{\uparrow}, u^{\downarrow}) = \frac{\sum_{i = 0}^{n} \max(i - \reluparam, 0) \frac{n!}{i! (n - i)!} \left( \frac{u^{\uparrow} \mathbb{P}(D \geq \kappa) + u^{\downarrow}}{n} \right)^{i} \left( 1 - \frac{u^{\uparrow} \mathbb{P}(D \geq \kappa) + u^{\downarrow}}{n} \right)^{n - i}}{u^{\uparrow} - \reluparam}
    \end{align*}
    is an expectation of a weakly increasing function with respect to a binomial distribution with mean $u^{\uparrow} \mathbb{P}(D \geq \kappa) + u^{\downarrow}$, and binomial distributions with higher means first-order stochastically dominate binomial distributions with lower means. Thus, $\frac{\partial G}{\partial u^{\downarrow}}(u^{\uparrow}, u^{\downarrow}(u^{\uparrow})) \geq 0$, and this fact together with
    \begin{align*}
        \frac{\partial G}{\partial u^{\downarrow}}(u^{\uparrow}, u^{\downarrow}(u^{\uparrow})) \frac{du^{\downarrow}(u^{\uparrow})}{du^{\uparrow}} \geq 0
    \end{align*}
    lets us conclude that
    \begin{align*}
         \frac{du^{\downarrow}(u^{\uparrow})}{du^{\uparrow}} \geq 0 \text{.}
    \end{align*}
    Since this result holds so long as we selected a $u^{\uparrow} \in (\reluparam, \mu^{\uparrow}]$, and $\kappa \geq \kappa - \frac{f(\kappa)}{f(\kappa^{+})} = \reluparam$, then as we move downward from $u^{\uparrow} = \mu^{\uparrow}$ to $u^{\uparrow} \to \kappa$ we obtain
    \begin{align*}
        u^{\downarrow}(\mu^{\uparrow}) \geq u^{\downarrow}(u^{\uparrow})
    \end{align*}
    for all $u^{\uparrow} \in (\kappa, \mu^{\uparrow}]$. Recall additionally that $u^{\downarrow}(u^{\uparrow}) \in [-u^{\uparrow} \mathbb{P}(D \geq \kappa), (\mu^{\uparrow} - u^{\uparrow}) \mathbb{P}(D \geq \kappa) + \mu^{\downarrow}]$, which provides an upper bound on $u^{\downarrow}(\mu^{\uparrow})$:
    \begin{align*}
        u^{\downarrow}(\mu^{\uparrow}) \leq (\mu^{\uparrow} - \mu^{\uparrow}) \mathbb{P}(D \geq \kappa) + \mu^{\downarrow} = \mu^{\downarrow}
    \end{align*}
    Hence, $u^{\downarrow}(u^{\uparrow}) \leq \mu^{\downarrow}$ for all $u^{\uparrow} \in (\kappa, \mu^{\uparrow}]$. Furthermore, from the facts that $G(u^{\uparrow}, u^{\downarrow}(u^{\uparrow})) = G(\mu^{\uparrow}, \mu^{\downarrow})$ for all $u^{\uparrow} \in (\kappa, \mu^{\uparrow}]$ and that $G(u^{\uparrow}, \cdot)$ is weakly increasing, we can deduce that
    \begin{align*}
        \mathbb{P}(D \geq \kappa) &\leq \lim_{z \downarrow \kappa} G(z, u^{\downarrow}(z))
        \\ &\leq \lim_{z \downarrow \kappa} G(z, \mu^{\downarrow})
        \\ &= \frac{\sum_{i = 0}^{n} \max(i - \reluparam, 0) \frac{n!}{i! (n - i)!} \left( \frac{\kappa \throughput}{n} \right)^{i} \left( 1 - \frac{\kappa \throughput}{n} \right)^{n - i}}{\kappa - \reluparam}
    \end{align*}
    with the bound becoming infinitely weak when $\threshold = \reluparam$ and $\throughput > 0$. We can observe that the right-hand side can be written as an expectation with respect to a binomial distribution with mean $\kappa \throughput$, and so we have
    \begin{align*}
        \mathbb{P}(D \geq \kappa) \leq \frac{\mathbb{E}[\max(X - \reluparam, 0)]}{\kappa - \reluparam}
    \end{align*}
    for $X \sim \text{Binomial}(n, \kappa \throughput \frac{1}{n})$.

    We now examine the case where $\reluparam < 0$. Let us fix the parameter $u^{\downarrow} = \mu^{\downarrow}$ and observe what happens as we vary the parameter $u^{\uparrow}$. We first note that since $\mu^{\downarrow} \geq 0$, the inequalities
    \begin{align*}
        \mu \geq u^{\uparrow} \mathbb{P}(D \geq \kappa) + \mu^{\downarrow} \geq 0
    \end{align*}
    hold for any $u^{\uparrow} \in [0, \mu^{\uparrow}]$, and so the summation
    \begin{align*}
        \sum_{i = 0}^{n} \max(i - \reluparam, 0) \frac{n!}{i! (n - i)!} \left( \frac{u^{\uparrow} \mathbb{P}(D \geq \kappa) + \mu^{\downarrow}}{n} \right)^{i} \left( 1 - \frac{u^{\uparrow} \mathbb{P}(D \geq \kappa) + \mu^{\downarrow}}{n} \right)^{n - i}
    \end{align*}
    in the numerator of $G(u^{\uparrow}, \mu^{\downarrow})$ is an expectation of $\max(i - \reluparam, 0)$ with respect to a binomial distribution with mean $u^{\uparrow} \mathbb{P}(D \geq \kappa) + \mu^{\downarrow}$. Additionally, if $\reluparam < 0$, then $\max(i - \reluparam, 0) = i - \reluparam$ for all $i \in \{ 0, 1, ..., n \}$. Thus, we compute that
    \begin{align*}
        G(u^{\uparrow}, \mu^{\downarrow}) &= \frac{\sum_{i = 0}^{n} \max(i - \reluparam, 0) \frac{n!}{i! (n - i)!} \left( \frac{u^{\uparrow} \mathbb{P}(D \geq \kappa) + \mu^{\downarrow}}{n} \right)^{i} \left( 1 - \frac{u^{\uparrow} \mathbb{P}(D \geq \kappa) + \mu^{\downarrow}}{n} \right)^{n - i}}{u^{\uparrow} - \reluparam}
        \\ &= \frac{\sum_{i = 0}^{n} (i - \reluparam) \frac{n!}{i! (n - i)!} \left( \frac{u^{\uparrow} \mathbb{P}(D \geq \kappa) + \mu^{\downarrow}}{n} \right)^{i} \left( 1 - \frac{u^{\uparrow} \mathbb{P}(D \geq \kappa) + \mu^{\downarrow}}{n} \right)^{n - i}}{u^{\uparrow} - \reluparam}
        \\ &= \frac{u^{\uparrow} \mathbb{P}(D \geq \kappa) + \mu^{\downarrow} - \reluparam}{u^{\uparrow} - \reluparam} \text{.}
    \end{align*}
    Taking the partial derivative of $G$ with respect to $u^{\uparrow}$, we obtain
    \begin{align*}
        \frac{\partial G}{\partial u^{\uparrow}}(u^{\uparrow}, \mu^{\downarrow}) &= \frac{\mathbb{P}(D \geq \kappa)}{u^{\uparrow} - \reluparam} - \frac{u^{\uparrow} \mathbb{P}(D \geq \kappa) + \mu^{\downarrow} - \reluparam}{(u^{\uparrow} - \reluparam)^{2}}
        \\ &= \frac{u^{\uparrow} \mathbb{P}(D \geq \kappa) - \reluparam \mathbb{P}(D \geq \kappa) - u^{\uparrow} \mathbb{P}(D \geq \kappa) - \mu^{\downarrow} + \reluparam}{(u^{\uparrow} - \reluparam)^{2}}
        \\ &= \frac{\reluparam \mathbb{P}(D < \kappa) - \mu^{\downarrow}}{(u^{\uparrow} - \reluparam)^{2}} \leq 0 \text{.}
    \end{align*}
    Thus, when $\reluparam < 0$, then for any $u^{\uparrow} \in [0, \mu^{\uparrow}]$ we have $\frac{\partial G}{\partial u^{\uparrow}}(u^{\uparrow}, \mu^{\downarrow}) \leq 0$. Note also that the expressions we derived for $G(u^{\uparrow}, \mu^{\downarrow})$ and $\frac{\partial G}{\partial u^{\uparrow}}(u^{\uparrow}, \mu^{\downarrow})$ were continuous in $u^{\uparrow}$ for all $u^{\uparrow} \in [0, \mu^{\uparrow}]$. Thus, moving from $u^{\uparrow} = \mu^{\uparrow}$ to $u^{\uparrow} = \kappa$ will increase $G(\cdot, \mu^{\downarrow})$, as $\kappa \leq \mu^{\uparrow}$ and $\kappa \in [0, \mu^{\uparrow}]$. We conclude that
    \begin{align*}
        G(\mu^{\uparrow}, \mu^{\downarrow}) &\leq G(\kappa, \mu^{\downarrow})
        \\ &= \frac{\mathbb{E}[\max(X - \reluparam, 0)]}{\kappa - \reluparam}
    \end{align*}
    for $X \sim \text{Binomial}(n, \kappa \throughput \frac{1}{n})$, and since $\mathbb{P}(D \geq \kappa) \leq G(\mu^{\uparrow}, \mu^{\downarrow}) \leq G(\kappa, \mu^{\downarrow})$, we have
    \begin{align*}
        \mathbb{P}(D \geq \kappa) \leq \frac{\mathbb{E}[\max(X - \reluparam, 0)]}{\kappa - \reluparam}
    \end{align*}
    in the case where $\reluparam < 0$ as well as $\reluparam \geq 0$.

    For the final step of the proof, recall that $\reluparam = \threshold - \frac{f(\threshold)}{f'(\threshold^{+})}$, which implies $\kappa - \reluparam \geq 0$, and so we have $\kappa - \reluparam = \max(\kappa - \reluparam, 0)$. This allows us to directly apply \Cref{relu_is_better} to show that
    \begin{align*}
        \frac{\mathbb{E}[\max(X - \reluparam, 0)]}{\max(\kappa - \reluparam, 0)} \leq \frac{\mathbb{E}[f(X)]}{f(\kappa)} \text{.}
    \end{align*}
    Chaining the previous two inequalities together allows us to conclude that
    \begin{align*}
        \mathbb{P}(D \geq \kappa) \leq \frac{\mathbb{E}[f(X)]}{f(\kappa)} \text{.}
    \end{align*}
\end{proof}

Some functions $f$ in the above theorem will produce stronger bounds, while others will produce weaker bounds. The following lemma can be used to narrow the search space for functions $f$ that produce strong bounds:

\begin{lemma}\label{relu_is_better}
    Let $f: \mathbb{R} \rightarrow \mathbb{R}$ be weakly convex and weakly positive. Then, for any random variable $D$ and for any threshold $\threshold \geq 0$ such that $f$ is increasing past $\threshold$, there exists a scaled ReLU function whose expectation is at least as small as the expectation of $f(D)$ normalized by $f(\threshold)$. More formally, we have
    \begin{align*}
        \frac{\mathbb{E}[\max(D - \reluparam, 0)]}{\max(\threshold - \reluparam, 0)} \leq \frac{\mathbb{E}[f(D)]}{f(\threshold)}
    \end{align*}
    for $\reluparam \equiv \threshold - \frac{f(\threshold)}{f'(\threshold^{+})}$, where we define the subderivative $f'\left( \threshold^{+} \right) \equiv \lim_{z \downarrow \threshold} \frac{f(z) - f(\threshold)}{z - \threshold}$.
\end{lemma}

\begin{proof}
    Observe that $x \mapsto \frac{\max(x - \reluparam, 0)}{\max(\threshold - \reluparam, 0)}$ is the scaled ReLU function which is tangent to $f(x) / f(\kappa)$ at $\threshold$. Note that this scaled ReLU is equal to $f(x) / f(\kappa)$ at $\kappa$ and weakly smaller than $f(x) / f(\kappa)$ elsewhere. See \Cref{fig:relu-is-better-sketch} for an illustration.
\end{proof}

\begin{figure}[h]
    \center
    \includegraphics[scale=0.4]{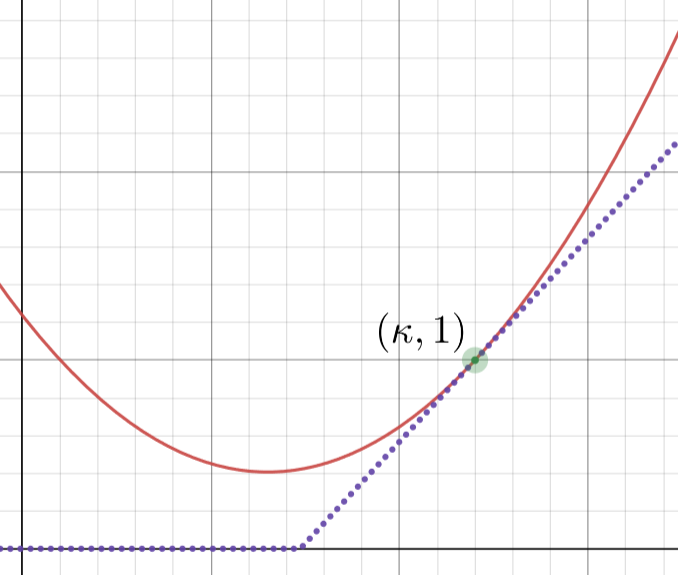}
    \caption{The solid red line is an example $\frac{f(x)}{f(\threshold)}$, and the dotted blue line is $\frac{\max(x - \reluparam, 0)}{\max(\threshold - \reluparam, 0)}$.}
    \label{fig:relu-is-better-sketch}
\end{figure}

\subsection{Analysis of Throughput and Availability}
This section sketches the proofs of \Cref{numerator_inequality}, \Cref{premarkov_tail_bound}, and \Cref{main_theorem}. Full proofs can be found in the appendix.

One way to conceptualize \Cref{numerator_inequality} is that it demonstrates a set of upper bounds on $\mathbb{E}[f(D)]$ that are tight when
\begin{enumerate}
    \item $D_{1}, ..., D_{n} \in [0, 1]$ place all their probability mass on $\{ 0, 1 \}$,
    \item $D_{1}, ..., D_{n}$ all allocate the same amount of probability mass to $0$ and to $1$, and
    \item the number $n$ is sent to infinity.
\end{enumerate}

To prove \Cref{numerator_inequality}, we first show the following lemma:

\begin{lemmarep}\label{sum_of_bernoullis}
    Let $f: \mathbb{R}_{+} \rightarrow \mathbb{R}$ be a weakly convex function, and let $D_{1}, ..., D_{n} \in [0, 1]$ be independent random variables. Additionally, let $Z_{1}, ..., Z_{n}$ be independently distributed Bernoulli random variables such that $\mathbb{E}[Z_{i}] = \mathbb{E}[D_{i}]$. Then, for all weakly positive translations $L \geq 0$,
    \begin{align*}
        \mathbb{E}[f(L + D)] \leq \mathbb{E}[f(L + Z)]
    \end{align*}
    where $D = \sum_{i = 1}^{n} D_{i}$ and $Z = \sum_{i = 1}^{n} Z_{i}$.
\end{lemmarep}

\begin{proof}
    As shorthand let $\mu_{1}, ..., \mu_{n}$ denote the means of $D_{1}, ..., D_{n}$. We will show that
    \begin{align*}
        \mathbb{E}[f(L + D)] \leq \sum_{S \subseteq [n]} f(L + |S|) \left[ \prod_{i \in S} \mu_{i} \right] \left[ \prod_{i \in [n] \setminus S} \left( 1 - \mu_{i} \right) \right] \text{.}
    \end{align*}
    by induction on $n$. Since the right hand side of the above is the formula for the the expectation $\mathbb{E}[f(L + Z)]$, this will complete the lemma. We will also demonstrate that the inequality in this lemma is tight when $D_{1}, ..., D_{n}$ are Bernoulli distributed.
    
    First, consider the base case where $n = 1$, and let any $L \geq 0$ be given. Note that since $n = 1$, we have $D = D_{1}$ and so $D \in [0, 1]$. Furthermore, note that if the mapping $x \mapsto f(x)$ is convex over the domain of non-negative real numbers, then the translated mapping $x \mapsto f(L + x)$ is also convex over the domain of non-negative real numbers. From these two facts, we then compute that
    \begin{align*}
        &\mathbin{\phantom{=}} \mathbb{E}[f(L + D)]
        \\ &= \mathbb{E}\left[ f\left( L + D + \left( 1 - D \right) \cdot 0 \right) \right]
        \\ &\leq \mathbb{E}\left[ D f\left( L + 1 \right) + \left( 1 - D \right) f\left( L + 0 \right) \right]
        \\ &= \mu_{1}f(L + 1) + \left( 1 - \mu_{1} \right) f(L)
        \\ &= \sum_{S \subseteq [1]} f(L + |S|) \left[ \prod_{i \in S} \mu_{i} \right] \left[ \prod_{i \in [1] \setminus S} \left( 1 - \mu_{i} \right) \right] \text{.}
    \end{align*}
    Note that the one inequality in this series of steps is tight when $D_{1} \in \{ 0, 1 \}$ w.p. $1$, which for fixed mean $\mathbb{E}[D_{1}] = \mu_{1}$ occurs if and only if $D_{1} \sim \text{Bernoulli}(\mu_{1})$.

    Now, for the inductive step, let us assume that for all $L \geq 0$,
    \begin{align*}
        \mathbb{E}\left[ f\left( L + \sum_{i = 1}^{n} D_{i} \right) \right] \leq \sum_{S \subseteq [n]} f(L + |S|) \left[ \prod_{i \in S} \mu_{i} \right] \left[ \prod_{i \in [n] \setminus S} \left( 1 - \mu_{i} \right) \right]
    \end{align*}
    and that this inequality is tight when all random variables in the set $\{ D_{i} \}_{i = 1}^{n}$ are Bernoulli distributed. We will demonstrate that
    \begin{align*}
        \mathbb{E}\left[ f\left( L + \sum_{i = 1}^{n + 1} D_{i} \right) \right] \leq \sum_{S \subseteq [n + 1]} f(L + |S|) \left[ \prod_{i \in S} \mu_{i} \right] \left[ \prod_{i \in [n + 1] \setminus S} \left( 1 - \mu_{i} \right) \right]
    \end{align*}
    for all $L \geq 0$ and this inequality is also tight when the random variables in the set $\{ D_{i} \}_{i = 1}^{n + 1}$ are Bernoulli distributed.
    
    First, let an arbitrary $L \geq 0$ be given. Using the law of iterated expectations, we have that
    \begin{align*}
        \mathbb{E}\left[ f\left( L + \sum_{i = 1}^{n + 1} D_{i} \right) \right] = \mathbb{E}\left[ \mathbb{E}\left[ f\left( L + D_{n + 1} + \sum_{i = 1}^{n} D_{i} \right) \mid D_{n + 1} \right] \right] \text{.}
    \end{align*}
    Since $D_{1}, ..., D_{n + 1}$ are independent, $L + D_{n + 1}$ can be treated as the translation in our inductive hypothesis, i.e.
    \begin{align*}
        &\mathbin{\phantom{=}} \mathbb{E}\left[ f\left( L + D_{n + 1} + \sum_{i = 1}^{n} D_{i} \right) \mid D_{n + 1} \right]
        \\ &\leq \sum_{S \subseteq [n]} f(L + D_{n + 1} + |S|) \left[ \prod_{i \in S} \mu_{i} \right] \left[ \prod_{i \in [n] \setminus S} \left( 1 - \mu_{i} \right) \right]
    \end{align*}
    and recall that by our inductive hypothesis, this inequality is tight when $\{ D_{i} \}_{i = 1}^{n}$ are Bernoulli distributed. Inserting this result into our earlier nested expectation, we have
    \begin{align*}
        &\mathbin{\phantom{=}} \mathbb{E}\left[ f\left( L + \sum_{i = 1}^{n + 1} D_{i} \right) \right]
        \\ &= \mathbb{E}\left[ \mathbb{E}\left[ f\left( L + D_{n + 1} + \sum_{i = 1}^{n} D_{i} \right) \mid D_{n + 1} \right] \right]
        \\ &\leq \mathbb{E}\left[ \sum_{S \subseteq [n]} f(L + D_{n + 1} + |S|) \left[ \prod_{i \in S} \mu_{i} \right] \left[ \prod_{i \in [n] \setminus S} \left( 1 - \mu_{i} \right) \right] \right]
        \\ &= \sum_{S \subseteq [n]} \mathbb{E}[f(L + D_{n + 1} + |S|)] \left[ \prod_{i \in S} \mu_{i} \right] \left[ \prod_{i \in [n] \setminus S} \left( 1 - \mu_{i} \right) \right]
        \\ &\leq \sum_{S \subseteq [n]} \left[ \mu_{n + 1} f(L + (|S| + 1)) + \left( 1 - \mu_{n + 1} \right) f(L + |S|) \right] \left[ \prod_{i \in S} \mu_{i} \right] \left[ \prod_{i \in [n] \setminus S} \left( 1 - \mu_{i} \right) \right]
        \\ &= \sum_{S \subseteq [n]} f(L + (|S| + 1)) \left[ \prod_{i \in S \cup \{ n + 1 \}} \mu_{i} \right] \left[ \prod_{i \in [n] \setminus S} \left( 1 - \mu_{i} \right) \right]
        \\ &\mathbin{\phantom{=}} + \sum_{S \subseteq [n]} f(L + |S|) \left[ \prod_{i \in S} \mu_{i} \right] \left[ \prod_{i \in [n + 1] \setminus S} \left( 1 - \mu_{i} \right) \right] \text{.}
        \intertext{Note that the second inequality on the lines above is tight when $D_{n + 1} \sim \text{Bernoulli}(\mu_{n + 1})$. Furthermore, the sums on the last line can be rewritten as}
        &\mathbin{\phantom{=}} \sum_{S \subseteq [n]} f(L + |S \cup \{ n + 1 \}|) \left[ \prod_{i \in S \cup \{ n + 1 \}} \mu_{i} \right] \left[ \prod_{i \in [n + 1] \setminus (S \cup \{ n + 1 \})} \left( 1 - \mu_{i} \right) \right]
        \\ &\mathbin{\phantom{=}} + \sum_{S \subseteq [n]} f(L + |S|) \left[ \prod_{i \in S} \mu_{i} \right] \left[ \prod_{i \in [n + 1] \setminus S} \left( 1 - \mu_{i} \right) \right]
        \\ &= \sum_{S \subseteq [n + 1]} \mathds{1}(n + 1 \in S) f(L + |S|) \left[ \prod_{i \in S} \mu_{i} \right] \left[ \prod_{i \in [n + 1] \setminus S} \left( 1 - \mu_{i} \right) \right]
        \\ &\mathbin{\phantom{=}} + \sum_{S \subseteq [n + 1]} \mathds{1}(n + 1 \notin S) f(L + |S|) \left[ \prod_{i \in S} \mu_{i} \right] \left[ \prod_{i \in [n + 1] \setminus S} \left( 1 - \mu_{i} \right) \right]
        \\ &= \sum_{S \subseteq [n + 1]} f(L + |S|) \left[ \prod_{i \in S} \mu_{i} \right] \left[ \prod_{i \in [n + 1] \setminus S} \left( 1 - \mu_{i} \right) \right] \text{.}
    \end{align*}
    This concludes the inductive step.
\end{proof}

When $L$ is set equal to $0$ in \Cref{sum_of_bernoullis}, it shows that the expectation of $f(D)$ is maximized when all $D_{i}$s are Bernoulli distributed random variables. However, to show \Cref{numerator_inequality}, we must show that it is further maximized when all such Bernoulli random variables are identically distributed, i.e. have the same mean. We demonstrate this in \Cref{iid_is_worst}.

\begin{lemmarep}\label{iid_is_worst}
    Let $f: \mathbb{R}_{+} \rightarrow \mathbb{R}$ be a weakly convex function, and let $D_{1}, ..., D_{n} \in \{ 0, 1 \}$ be independently distributed Bernoulli random variables with means $\mu_{i} = \mathbb{E}[D_{i}]$. Denote their sums $D = \sum_{i = 1}^{n} D_{i}$ and $\mu = \sum_{i = 1}^{n} \mu_{i}$. Then,
    \begin{align*}
        \mathbb{E}[f(D)] \leq \mathbb{E}[f(X)]
    \end{align*}
    for $X \sim \text{Binomial}(n, \mu / n)$. This inequality is tight when $D_{1}, ..., D_{n}$ are identically distributed, meaning $D \sim X$.
\end{lemmarep}

\begin{proof}
    We begin by introducing some notation. Let $\vec{\mu} \triangleq (\mu_{1}, ..., \mu_{n})$ and let
    \begin{align*}
        G(\vec{\mu}) \triangleq \sum_{S \subseteq [n]} f(|S|) \left[ \prod_{i \in S} \mu_{i} \right] \left[ \prod_{i \in [n] \setminus S} \left( 1 - \mu_{i} \right) \right] \text{.}
    \end{align*}
    
    Consider that the $2^{n}$ sets $S \subseteq [n]$ can be grouped into $2^{n - 1}$ pairs by way of the following rule: fixing a $k \in [n]$, two sets $S, S'$ are in a pair if they contain identical elements except that $k \in S$ and $k \notin S'$. We can then rewrite $G(\vec{\mu})$ to sum over these pairs $(S, S \setminus \{ k \})$ for all $S \subseteq [n]$ where $k \in S$:
    \begin{align*}
        G(\vec{\mu}) &= \sum_{S \subseteq [n]} f(|S|) \left[ \prod_{i \in S} \mu_{i} \right] \left[ \prod_{i \in [n] \setminus S} \left( 1 - \mu_{i} \right) \right]
        \\ &= \sum_{S \subseteq [n]} \mathds{1}(k \in S) f(|S|) \left[ \prod_{i \in S} \mu_{i} \right] \left[ \prod_{i \in [n] \setminus S} \left( 1 - \mu_{i} \right) \right]
        \\ &\mathbin{\phantom{=}} + \mathds{1}(k \in S) f(|S \setminus \{ k \}|) \left[ \prod_{i \in S \setminus \{ k \}} \mu_{i} \right] \left[ \prod_{i \in [n] \setminus (S \setminus \{ k \})} \left( 1 - \mu_{i} \right) \right]
        \\ &= \sum_{S \subseteq [n]} \mathds{1}(k \in S) f(|S|) \mu_{k} \left[ \prod_{i \in S \setminus \{ k \}} \mu_{i} \right] \left[ \prod_{i \in [n] \setminus S} \left( 1 - \mu_{i} \right) \right]
        \\ &\mathbin{\phantom{=}} + \mathds{1}(k \in S) f(|S| - 1) \left( 1 - \mu_{k} \right) \left[ \prod_{i \in S \setminus \{ k \}} \mu_{i} \right] \left[ \prod_{i \in [n] \setminus S} \left( 1 - \mu_{i} \right) \right]
        \\ &= \sum_{S \subseteq [n]} \mathds{1}(k \in S) \left( \mu_{k} (f(|S|) - f(|S| - 1)) + f(|S| - 1) \right) \left[ \prod_{i \in S \setminus \{ k \}} \mu_{i} \right] \left[ \prod_{i \in [n] \setminus S} \left( 1 - \mu_{i} \right) \right]
    \end{align*}
    
    Rewriting $G(\vec{\mu})$ this way permits us to easily compute its partial derivative with respect to $\mu_{k}$:
    \begin{align*}
        \frac{\partial G}{\partial \mu_{k}}(\vec{\mu}) &= \sum_{S \subseteq [n]} \mathds{1}(k \in S) \left( f(|S|) - f(|S| - 1) \right) \left[ \prod_{i \in S \setminus \{ k \}} \mu_{i} \right] \left[ \prod_{i \in [n] \setminus S} \left( 1 - \mu_{i} \right) \right]
    \end{align*}

    As shorthand, let
    \begin{align*}
        \Delta f(|S|) \triangleq f(|S|) - f(|S| - 1)
    \end{align*}
    and observe that since $f$ is weakly concave, $\Delta f(\cdot)$ is weakly increasing. Furthermore, we observe that $\frac{\partial G}{\partial \mu_{k}}(\vec{\mu})$ can be written as a conditional expectation: if we define $n$ independent random variables $X_{1, \mu_{1}}, ..., X_{n, \mu_{n}}$ such that $X_{i, \mu_{i}} \sim \text{Bernoulli}(\mu_{i})$, then an equivalent expression for $\frac{\partial G}{\partial \mu_{k}}(\vec{\mu})$ is
    \begin{align*}
        \frac{\partial G}{\partial \mu_{k}}(\vec{\mu}) = \mathbb{E}\left[ \Delta f\left( \sum_{i = 1}^{n} X_{i, \mu_{i}} \right) \mid X_{k, \mu_{k}} = 1 \right] = \mathbb{E}\left[ \Delta f\left( 1 + \sum_{i \neq k} X_{i, \mu_{i}} \right) \right] \text{.}
    \end{align*}
    We can then compare sizes of different partial derivatives of $G$ by comparing these expectations. For example, selecting $k, \ell \in [n]$ such that $k \neq \ell$, we have
    \begin{align*}
        &\mathbin{\phantom{=}} \frac{\partial G}{\partial \mu_{k}}(\vec{\mu}) - \frac{\partial G}{\partial \mu_{\ell}}(\vec{\mu})
        \\ &= \mathbb{E}\left[ \Delta f\left( 1 + \sum_{i \neq k} X_{i, \mu_{i}} \right) \right] - \mathbb{E}\left[ \Delta f\left( 1 + \sum_{i \neq \ell} X_{i, \mu_{i}} \right) \right]
        \\ &= \mathbb{E}\left[ \Delta f\left( 1 + \sum_{i \neq k} X_{i, \mu_{i}} \right) - \Delta f\left( 1 + \sum_{i \neq k} X_{i, \mu_{i}} \right) \right]
        \\ &= \mathbb{E}\left[ \Delta f\left( X_{\ell, \mu_{\ell}} + 1 + \sum_{i \neq k, \ell} X_{i, \mu_{i}} \right) - \Delta f\left( X_{k, \mu_{k}} + 1 + \sum_{i \neq k, \ell} X_{i, \mu_{i}} \right) \right]
        \\ &= \mathbb{E}\left[ \mathbb{E}\left[ \Delta f\left( X_{\ell, \mu_{\ell}} + 1 + \sum_{i \neq k, \ell} X_{i, \mu_{i}} \right) - \Delta f\left( X_{k, \mu_{k}} + 1 + \sum_{i \neq k, \ell} X_{i, \mu_{i}} \right) \mid \sum_{i \neq k, \ell} X_{i, \mu_{i}} \right] \right]
        \\ &= \mathbb{E}[ \mu_{\ell} \Delta f\left( 2 + \sum_{i \neq k, \ell} X_{i, \mu_{i}} \right) + \left( 1 - \mu_{\ell} \right) \Delta f\left( 1 + \sum_{i \neq k, \ell} X_{i, \mu_{i}} \right)
        \\ &\mathbin{\phantom{=}} - \mu_{k} \Delta f\left( 2 + \sum_{i \neq k, \ell} X_{i, \mu_{i}} \right) - \left( 1 - \mu_{k} \right) \Delta f\left( 1 + \sum_{i \neq k, \ell} X_{i, \mu_{i}} \right) ]
        \\ &= \mathbb{E}\left[ \Delta f\left( 2 + \sum_{i \neq k, \ell} X_{i, \mu_{i}} \right) \left( \mu_{\ell} - \mu_{k} \right) - \Delta f\left( 1 + \sum_{i \neq k, \ell} X_{i, \mu_{i}} \right) \left( \mu_{\ell} - \mu_{k} \right) \right]
        \\ &= \mathbb{E}\left[ \left( \Delta f\left( 2 + \sum_{i \neq k, \ell} X_{i, \mu_{i}} \right) - \Delta f\left( 1 + \sum_{i \neq k, \ell} X_{i, \mu_{i}} \right) \right) \left( \mu_{\ell} - \mu_{k} \right) \right]
        \\ &= \mathbb{E}\left[ \Delta f\left( 2 + \sum_{i \neq k, \ell} X_{i, \mu_{i}} \right) - \Delta f\left( 1 + \sum_{i \neq k, \ell} X_{i, \mu_{i}} \right) \right] \left( \mu_{\ell} - \mu_{k} \right) \text{.}
    \end{align*}

    Note that since $\Delta f (\cdot)$ is weakly increasing, the term
    \begin{align*}
        \mathbb{E}\left[ \Delta f\left( 2 + \sum_{i \neq k, \ell} X_{i, \mu_{i}} \right) - \Delta f\left( 1 + \sum_{i \neq k, \ell} X_{i, \mu_{i}} \right) \right]
    \end{align*}
    is always weakly positive. Thus, if $\mu_{\ell} \geq \mu_{k}$, then $\frac{\partial G}{\partial \mu_{k}}(\vec{\mu}) \geq \frac{\partial G}{\partial \mu_{\ell}}(\vec{\mu})$. Or, more colloquially, \textit{the highest partial derivatives of $G$ are those taken with respect to the input coordinate with the lowest-valued argument}.
    
    We now define the vector-valued function $\vec{u} : [0, 1] \rightarrow [0, 1]^{n}$ such that
    \begin{align*}
        \vec{u}(t) = (1 - t) \vec{\mu} + t \frac{\mu}{n} \vec{1}_{n}
    \end{align*}
    where $\vec{1}_{n} \in \mathbb{R}^{n}$ denotes a vector with a $1$ in every component. In other words, the function $\vec{u}$ defines a straight path between a vector of the original $\mu_{i}$s at $t = 0$ and a vector of identical means at $t = 1$, such that the sum
    \begin{align*}
        \sum_{i = 1}^{n} \vec{u}_{i}(t) = \mu
    \end{align*}
    stays constant for all $t \in [0, 1]$. Note that the order of the components of $\vec{u}(t)$ are the same as the order of the components of $\vec{\mu}$, i.e. if $\mu_{i} \geq \mu_{j}$, then for all $t \in [0, 1]$, $\vec{u}_{i}(t) \geq \vec{u}_{j}(t)$.

    Plugging $\vec{u}(t)$ into our objective function $G$, we obtain
    \begin{align*}
        G(\vec{u}(t)) = \sum_{S \subseteq [n]} f(|S|) \left[ \prod_{i \in S} \vec{u}_{i}(t) \right] \left[ \prod_{i \in [n] \setminus S} \left( 1 - \vec{u}_{i}(t) \right) \right] \text{,}
    \end{align*}
    and taking the first derivative of $G(\vec{u}(t))$ with respect to $t$ produces
    \begin{align*}
        \frac{d(G(\vec{u}(t)))}{dt} &= \sum_{i = 1}^{n} \frac{\partial G}{\partial \mu_{i}}(\vec{u}(t)) \frac{d(\vec{u}_{i}(t))}{dt}
        \\ &= \sum_{i = 1}^{n} \frac{\partial G}{\partial \mu_{i}}(\vec{u}(t)) \left( \frac{\mu}{n} - \mu_{i} \right) \text{.}
    \end{align*}
    If $\frac{\mu}{n} - \mu_{i} \leq \frac{\mu}{n} - \mu_{j}$, then $\mu_{i} \geq \mu_{j}$. Additionally, recall that if $\mu_{i} \geq \mu_{j}$ then $\vec{u}_{i}(t) \geq \vec{u}_{j}(t)$, which in turn implies $\frac{\partial G}{\partial \mu_{i}}(\vec{u}(t)) \leq \frac{\partial G}{\partial \mu_{j}}(\vec{u}(t))$. Therefore, $\frac{\partial G}{\partial \mu_{i}}(\vec{u}(t))$ and $\frac{\mu}{n} - \mu_{i}$ are monotonically related in the selection of $i \in [n]$, and thus
    \begin{align*}
        n \frac{1}{n} \sum_{i = 1}^{n} \frac{\partial G}{\partial \mu_{i}}(\vec{u}(t)) \left( \frac{\mu}{n} - \mu_{i} \right) &\geq n \left( \frac{1}{n} \sum_{i = 1}^{n} \frac{\partial G}{\partial \mu_{i}}(\vec{u}(t)) \right) \left( \frac{1}{n} \sum_{i = 1}^{n} \left( \frac{\mu}{n} - \mu_{i} \right) \right)
        \\ &= \left( \frac{1}{n} \sum_{i = 1}^{n} \frac{\partial G}{\partial \mu_{i}}(\vec{u}(t)) \right) (\mu - \sum_{i = 1}^{n} \mu_{i}) = 0 \text{.}
    \end{align*}
    We can therefore conclude that
    \begin{align*}
        \frac{d(G(\vec{u}(t)))}{dt} \geq 0
    \end{align*}
    for all $t \in [0, 1]$, and thus the path defined by $\vec{u}(\cdot)$ represents a way to weakly increase $G$ by moving from $t = 0$ to $t = 1$, i.e.
    \begin{align*}
        G\left( \frac{\mu}{n} \vec{1}_{n} \right) = G(\vec{u}(1)) \geq G(\vec{u}(0)) = G(\vec{\mu})
    \end{align*}
    and thus by the definition of $G$,
    \begin{align*}
        \sum_{S \subseteq [n]} f(|S|) \left[ \prod_{i \in S} \mu_{i} \right] \left[ \prod_{i \in [n] \setminus S} \left( 1 - \mu_{i} \right) \right] \leq \sum_{S \subseteq [n]} f(|S|) \left[ \prod_{i \in S} \frac{\mu}{n} \right] \left[ \prod_{i \in [n] \setminus S} \left( 1 - \frac{\mu}{n} \right) \right] \text{.}
    \end{align*}
    Observe that the above inequality can be rewritten as expectations with respect to $D$ and a variant of $D$ composed of identically distributed Bernoulli random variables, i.e. a binomial distribution. In other words,
    \begin{align*}
        \mathbb{E}[f(D)] \leq \mathbb{E}[f(X)]
    \end{align*}
    for $X \sim \text{Binomial}(n, \mu / n)$, which concludes the proof.
\end{proof}

A proof of \Cref{iid_is_worst} can also be found in other works \parencite{hoeffding-56, AM-23}, and can be naturally extended to show that $\mathbb{E}[f(X)]$ is maximized as the number of demands $n \to \infty$, i.e. as the binomial $X$ approaches a Poisson $Y$ with mean $\mathbb{E}[Y] = \mathbb{E}[X]$. The fact that $\mathbb{E}[f(X)] \leq \mathbb{E}[f(Y)]$ is used to conclude \Cref{numerator_inequality}. Furthermore, observe that $\mathbb{E}[f(X)] \leq \mathbb{E}[f(Y)]$ additionally implies that the bound in \Cref{main_theorem} is weakest when $n \to \infty$.

The proof of \Cref{premarkov_tail_bound} from \Cref{numerator_inequality} is relatively straightforward. It follows a similar structure to Markov's inequality on the random variable $f(D)$.

\begin{proof}[Proof of \Cref{premarkov_tail_bound}]
    From the law of total expectation, we have
    \begin{align*}
        \mathbb{E}[f(D)] = \mathbb{E}[f(D) \, \mathds{1}(D < \kappa)] + \mathbb{E}[f(D) \mid D \geq \threshold] \mathbb{P}(D \geq \threshold) \text{.}
    \end{align*}
    Since $f$ is weakly positive, $\mathbb{E}[f(D) \, \mathds{1}(D < \kappa)] \geq 0$. Furthermore, since $f$ is weakly convex, $\mathbb{E}[f(D) \mid D \geq \threshold] \geq f(\mathbb{E}[D \mid D \geq \threshold])$. Thus,
    \begin{align*}
         f(\mathbb{E}[D \mid D \geq \threshold]) \, \mathbb{P}(D \geq \threshold) \leq \mathbb{E}[f(D)] \text{.}
    \end{align*}
    Next, from \Cref{numerator_inequality}, we have $\mathbb{E}[f(D)] \leq \mathbb{E}[f(X)]$. Therefore,
    \begin{align*}
        f(\mathbb{E}[D \mid D \geq \threshold]) \, \mathbb{P}(D \geq \threshold) \leq \mathbb{E}[f(X)] \text{.}
    \end{align*}
    Dividing both sides by $f(\mathbb{E}[D \mid D \geq \threshold]) \geq 0$, we obtain as desired:
    \begin{align*}
        \mathbb{P}(D \geq \threshold) &\leq \frac{\mathbb{E}[f(X)]}{f(\mathbb{E}[D \mid D \geq \threshold])} \text{.} \qedhere
    \end{align*}
\end{proof}

Observe that the most notable difference between \Cref{premarkov_tail_bound} and Markov's inequality is that the denominator on the right-hand side is $f(\mathbb{E}[D \mid D \geq \threshold])$, not $f(\threshold)$.

Proving \Cref{main_theorem} from \Cref{premarkov_tail_bound} is quite involved.

\begin{proof}[Proof sketch of \Cref{main_theorem}]
    We select the ReLU mapping $f(x) = \max(x - \reluparam, 0)$ for a value of $\reluparam \leq \threshold$, and obtain
    \begin{align*}
        \mathbb{P}(D \geq \threshold) \leq \frac{\mathbb{E}[\max(X - \reluparam, 0)]}{\max(\mathbb{E}[D \mid D \geq \threshold] - \reluparam, 0)}
    \end{align*}
    from \Cref{premarkov_tail_bound}, with $X \sim \text{Binomial}(n, \mu / n)$. Observe that since expectations with respect to a binomial distribution have a closed form, the right-hand side can be written as a function of two parameters, $\mu^{\uparrow} \triangleq \mathbb{E}[D \mid D \geq \threshold]$ and $\mu^{\downarrow} \triangleq \mathbb{E}[D \mathds{1}(D < \threshold)]$, as follows:
    \begin{align*}
        \frac{\mathbb{E}[\max(X - \reluparam, 0)]}{\max(\mathbb{E}[D \mid D \geq \threshold] - \reluparam, 0)} = \frac{\sum_{i = 0}^{n} \max(i - \reluparam, 0) \frac{n!}{i! (n - i)!} \left( \frac{\mu^{\uparrow} \mathbb{P}(D \geq \threshold) + \mu^{\downarrow}}{n} \right)^{i} \left( 1 - \frac{\mu^{\uparrow} \mathbb{P}(D \geq \threshold) + \mu^{\downarrow}}{n} \right)^{n - i}}{\max(\mu^{\uparrow} - \reluparam, 0)}
    \end{align*}

    In the case where $\reluparam < 0$, we show directly that the right-hand side is decreasing in $\mu^{\uparrow}$, and move $\mu^{\uparrow} \to \threshold$. In the more complicated case where $\reluparam \geq 0$, we adjust the parameters of the right-hand side in two steps:
    \begin{enumerate}
        \item Move along a path of pairs $(\mu^{\uparrow}, \mu^{\downarrow})$ that keeps the right-hand side constant, and such that the end of the path is the point $(\threshold, t)$ for some $t$ below the original $\mu^{\downarrow}$.
        \item Adjust $t$ upwards to the original $\mu^{\downarrow}$, and show that this only increases the right-hand side.
    \end{enumerate}
    Therefore, the final value of the right-hand side should be at least $\mathbb{P}(D \geq \threshold)$. See \Cref{fig:main-theorem-proof-sketch} for an illustration.
    
    \begin{figure}[h]
        \center
        \includegraphics[scale=0.4]{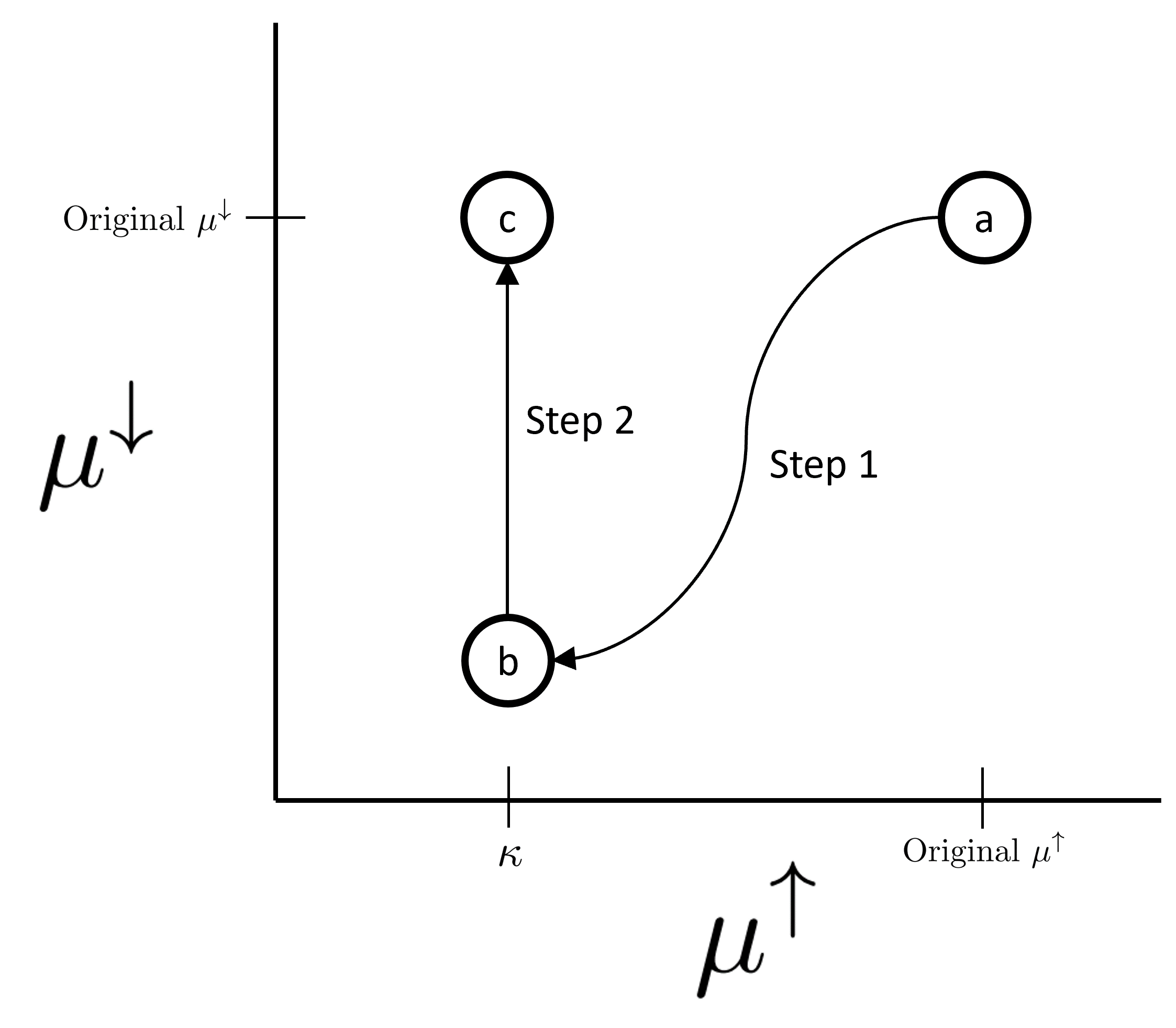}
        \caption{Step 1 moves between different $(\mu^{\uparrow}, \mu^{\downarrow})$ pairs. At all points along this path in step 1, the RHS stays constant and thus above $\mathbb{P}(D \geq \threshold)$. Then, since the RHS is increasing along the $\mu^{\downarrow}$ axis, moving from $b$ to $c$ in step 2 raises the RHS even further above $\mathbb{P}(D \geq \threshold)$.}
        \label{fig:main-theorem-proof-sketch}
    \end{figure}

    The resulting formula for the right-hand side is equal to the formula for the expectation of a binomial distribution with mean $\threshold \mathbb{P}(D \geq \threshold) + \mu^{\downarrow} = \threshold \throughput$, divided by the value $\max(\threshold - \reluparam, 0)$. A straightforward application of \Cref{relu_is_better} completes the theorem, showing that this ReLU bound can be weakened to $f(X) / f(\threshold)$ for $X \sim \text{Binomial}(n, \threshold \throughput / n)$.
\end{proof}

Observe that the ReLU function described in \Cref{relu_is_better} is weakly convex, weakly positive, and is strictly increasing past $\threshold$. It thus satisfies all the requirements for making a concentration inequality using \Cref{main_theorem}, and we can conclude that for any concentration inequality made using \Cref{main_theorem}, there exists a ReLU function that produces a concentration inequality at least as strong. This means that in order to find the functions $f$ for use with \Cref{main_theorem} that produce the strongest concentration inequalities, it suffices to search solely among the class of ReLU functions.

\section{Multi-Unit Posted-Price Mechanisms}
In this section we apply our analysis of throughput and availability to deriving prophet inequalities for posted-price mechanisms. Typically, prophet inequalities are concerned with showcasing the existence of a price that obtains good worst-case expected welfare. However, a price that obtains good worst-case expected welfare can potentially result in low availability. We examine what occurs when the mechanism designer cares about both availability and expected welfare, converting the availability-throughput tradeoff curves from the soup kitchen problem into availability-welfare tradeoff curves along which a mechanism designer can select their preferred location.

\subsection{Model}
Consider running a posted-price mechanism to sell $K \in \mathbb{R}_{+}$ units of a divisible good to $n$ agents. Agents purchase up to a unit of the good and agent $i$'s value for receiving $x$ units of good is given by the value function $v_{i}(x)$, $v_{i}: [0, 1] \xrightarrow{} \mathbb{R}$, drawn from some distribution $F_{i}$ over value functions. All distributions $F_{i}$ are independent and we assume $v_i(0) = 0$.

Agents have quasi-linear utility from purchasing $x$ units of good at price $p$ per-unit good, i.e, agent $i$ gets a utility
\begin{align*}
    u_{i}(x) = v_{i}(x) - px \text{.}
\end{align*}
Thus, if $S_i$ units of the good are left unsold when the agent interacts with the mechanism, the agent purchases $B_i = \argmax_{x \in [0, \min \{1, S_i\}]} \{ v_i(x) - px\}$ to maximize its utility. We make a mild assumption on the valuation function $v_i$ that for any quantity of residual supply $S_i$, there is a maximizer $B_i$ that optimizes agent $i$'s utility (if this is not unique, the agent can mix between buying any of these quantities). Observe that the class of all weakly concave functions is contained in the class of functions satisfying the above assumption.

Let $D_i = \argmax_{x \in [0, 1]} \{ v_i(x) - px\}$ be agent $i$'s desired quantity of goods assuming no supply constraints, and let $D = \sum_{i = 1}^n D_i$ be the total quantity of goods demanded by the agents.

The agents' arrival order is determined adversarially as follows:
\begin{enumerate}
    \item Agents realize their value functions $(v_{1}, \dots, v_{n})$.
    \item An adversary determines the order in which the agents interact with the online mechanism fully knowing agents' value functions so as to minimize the welfare of the mechanism.
    \item Agents interact with the mechanism in the order determined by the adversary, purchasing $B_{i}$ units of the good so as to optimize their utility.
\end{enumerate}

The mechanism designer cares about selecting a price $p$, not only to achieve an allocation with high expected welfare, but to make sure the mechanism can reliably serve agents. We define
\begin{align*}
    \delta \triangleq \mathbb{P}\left( \bigcup_{i = 1}^{n} \{ D - D_{i} > K - 1 \} \right)
\end{align*}
as the probability that there exists some agent who, if placed at the worst/latest possible location in the arrival order, would be denied the option to purchase up to $1$ unit of the good.

$1 - \delta$ is closely related to the concept of availability $\availability$ introduced in the soup kitchen problem. A low $\delta$ means that shoppers can be confident that the posted-price mechanism will have enough of the good to serve them, no matter what they choose to buy. Observe also that if we define the supply threshold $\threshold = K - 1$ and the availability of the mechanism as $\availability = \mathbb{P}(D < \threshold)$, then the availability is a lower bound on $1 - \delta$. We will therefore term $\delta$ the "real unavailability" and $K$ the "real supply" when needed to differentiate from our use of unavailability $1 - \availability$ and supply $\threshold$ from the soup kitchen problem, though we will refer to both by the same term outside such cases.

Let $p(\delta)$ be the price such that, with probability $1 - \delta$, every agent is able to choose its desired quantity irrespective of its position in the queue. Note that $p(\delta)$ is independent of the order in which agents interact with the mechanism. We trace out a lower bound on the expected welfare $\text{APX}$ of the posted-price mechanism at price $p(\delta)$, against that of the optimal offline mechanism $\text{OPT}$, as a function of $\delta$.

\subsection{Results}
We produce a prophet inequality which demonstrates the tradeoff between $\delta$ and expected welfare.
\begin{propositionrep}\label{welfare_bound}
    In a multi-unit posted-price mechanism with up-to-unit demand buyers, where $\delta \triangleq \mathbb{P}\left( \bigcup_{i = 1}^{n} \{ D - D_{i} > K - 1 \} \right)$ is the probability that any agent, if prioritized last by the mechanism, does not have the option to buy $1$ unit of the good, then
    \begin{align*}
        \frac{\text{APX}}{\text{OPT}} \geq \min\left( \frac{K - 1}{K} \mathbb{E}\left[ \min\left(\frac{D}{K - 1}, 1 \right) \right], 1 - \delta \right) \text{.}
    \end{align*}
\end{propositionrep}

\begin{proof}
    Let $B = \sum_{i = 1}^{n} B_{i}$ denote the sum of the purchases of all $n$ agents. To compute the worst-case (lowest) value of $\text{APX} / \text{OPT}$, we begin by lower bounding $\text{APX}$, followed by upper bounding $\text{OPT}$. First, the lower bound for $\text{APX}$:

    \begin{align*}
        \text{APX} &= \mathbb{E}[\sum_{i = 1}^{n} v_{i}(B_{i})]
        \\ &= \mathbb{E}[\sum_{i = 1}^{n} \left( pB_{i} + (v_{i}(B_{i}) - pB_{i}) \right)]
        \\ &= p \mathbb{E}[B] + \sum_{i = 1}^{n} \mathbb{E}[v_{i}(B_{i}) - pB_{i}]
        \\ &= p \mathbb{E}[B] + \sum_{i = 1}^{n} \mathbb{E}[u_{i}(B_{i})]
        \\ &= p \mathbb{E}[B] + \sum_{i = 1}^{n} \mathbb{E}[u_{i}(B_{i}) \mathds{1}(S_{i} \geq M)] + \sum_{i = 1}^{n} \mathbb{E}[u_{i}(B_{i}) \mathds{1}(S_{i} < M)]
        \\ &= p \mathbb{E}[B] + \sum_{i = 1}^{n} \mathbb{E}[u_{i}(D_{i}) \mathds{1}(S_{i} \geq M)] + \sum_{i = 1}^{n} \mathbb{E}[u_{i}(B_{i}) \mathds{1}(S_{i} < M)]
        \\ &\geq p \mathbb{E}[B] + \sum_{i = 1}^{n} \mathbb{E}[u_{i}(D_{i}) \mathds{1}(S_{i} \geq M)]
    \end{align*}
    with the last step following from the fact that agents, no matter how much is left on the shelf when they arrive at the mechanism, can always choose to buy nothing at all and obtain 0 utility. More formally,
    \begin{align*}
        u_{i}(B_{i}) &= \max_{b \in [0, \min(1, S_{i})]} \left[ v_{i}(b) - pb \right]
        \\ &\geq \max_{b \in [0, \min(1, S_{i})]} \left[ v_{i}(b) - pb \right]
        \\ &= v_{i}(0) - p \cdot (0)
        \\ &= 0 \text{,}
    \end{align*}
    as by assumption, $v_{i}(0) = 0$.
    
    We can simplify this lower bound on $\text{APX}$ by noting that when $D - D_{i} \leq K - 1$, it implies that $S_{i} \geq 1$. Thus,
    
    \begin{align*}
        \text{APX} &\geq p \mathbb{E}[B] + \sum_{i = 1}^{n} \mathbb{E}[u_{i}(D_{i}) \mathds{1}(D - D_{i} \leq K - 1)]
        \\ &= p \mathbb{E}[B] + \sum_{i = 1}^{n} \mathbb{E}[u_{i}(D_{i}) \mid D - D_{i} \leq K - 1] \mathbb{P}(D - D_{i} \leq K - 1) \text{.}
    \end{align*}
    
    Since $D_{i}$ and $D - D_{i}$ are independent, $\mathbb{E}[u_{i}(D_{i}) \mid D - D_{i} \leq K - 1] = \mathbb{E}[u_{i}(D_{i})]$, and we have
    \begin{align*}
        \text{APX} \geq p \mathbb{E}[B] + \sum_{i = 1}^{n} \mathbb{E}[u_{i}(D_{i})] \mathbb{P}(D - D_{i} \leq K - 1) \text{.}
    \end{align*}
    
    Lastly, we will lower bound both $\mathbb{P}(D - D_{i} \leq K - 1)$ and $\mathbb{E}[u_{i}(D_{i})]$ and substitute these into our lower bound for $\text{APX}$. First, we compute
    \begin{align*}
        \mathbb{P}(D - D_{i} \leq K - 1) &\geq \mathbb{P}(\bigcap_{i = 1}^{n} \{ D - D_{i} \leq K - 1 \})
        \\ &= 1 - \mathbb{P}(\bigcup_{i = 1}^{n} \{ D - D_{i} > K - 1 \})
        \\ &= 1 - \delta
    \end{align*}
    and so $\text{APX} \geq p \mathbb{E}[B] + (1 - \delta) \sum_{i = 1}^{n} \mathbb{E}[u_{i}(D_{i})]$. Next, we note that since agent $i$ selects a $D_{i}$ to maximize their own utility, subject to the restriction that they purchase no more than $1$ unit of the good, the amount of the good they would receive in a socially optimal allocation subject to the same restriction would give them at most the same utility. More formally,
    \begin{align*}
        \mathbb{E}[u_{i}(D_{i})] &= \mathbb{E}[\mathbb{E}[u_{i}(D_{i}) \mid \vec{v}]]
        \\ &= \mathbb{E}[\mathbb{E}[\max_{b \in [0, 1]}[v_{i}(b) - pb] \mid \vec{v}]]
        \\ &\geq \mathbb{E}[\mathbb{E}[v_{i}(B_{i}^{*}) - pB_{i}^{*} \mid \vec{v}]]
        \\ &= \mathbb{E}[v_{i}(B_{i}^{*}) - pB_{i}^{*}]
        \\ &= \mathbb{E}[u_{i}(B_{i}^{*})] \text{.}
    \end{align*}
    
    Therefore, we can conclude that $\text{APX} \geq p \mathbb{E}[B] + (1 - \delta) \sum_{i = 1}^{n} \mathbb{E}[u_{i}(B_{i}^{*})]$.
    
    We will now proceed to upper bound $\text{OPT}$. We compute as follows:
    \begin{align*}
        \text{OPT} &= \mathbb{E}[\sum_{i = 1}^{n} v_{i}(B_{i}^{*})]
        \\ &= \mathbb{E}[\sum_{i = 1}^{n} (pB_{i}^{*} + (v_{i}(B_{i}^{*}) - pB_{i}^{*}))]
        \\ &= p \mathbb{E}[\sum_{i = 1}^{n} B^{*}_{i}] + \sum_{i = 1}^{n} \mathbb{E}[v_{i}(B_{i}^{*}) - pB_{i}^{*}]
        \\ &= p \mathbb{E}[\sum_{i = 1}^{n} B^{*}_{i}] + \sum_{i = 1}^{n} \mathbb{E}[u_{i}(B_{i}^{*})]
        \\ &\leq pK + \sum_{i = 1}^{n} \mathbb{E}[u_{i}(B_{i}^{*})]
    \end{align*}
    
    The last step follows from the fact that no allocation, including the optimal allocation, can use more than the total stock of $K$ goods.
    
    Combining our lower and upper bounds for $\text{APX}$ and $\text{OPT}$, respectively, we have
    \begin{align*}
        \frac{\text{APX}}{\text{OPT}} &\geq \frac{p \mathbb{E}[B] + (1 - \delta) \sum_{i = 1}^{n} \mathbb{E}[u_{i}(B_{i}^{*})]}{pK + \sum_{i = 1}^{n} \mathbb{E}[u_{i}(B_{i}^{*})]}
        \\ &\geq \min\left( \frac{\mathbb{E}[B]}{K}, 1 - \delta \right) \text{.}
    \end{align*}
    
    This splits our welfare bound into two cases; a high $\delta$ case, where the term $1 - \delta$ is smaller, and a low $\delta$ case, where the term $\mathbb{E}[B] / K$ is smaller.
    
    Our goal now becomes finding a lower bound for $\mathbb{E}[B] / K$ in terms of $D$. We begin by applying the law of total expectation to $\mathbb{E}[B]$, yielding
    \begin{align*}
         \mathbb{E}[B] &= \mathbb{E}[B \mid \bigcap_{i = 1}^{n} \{ S_{i} \geq 1 \}] \mathbb{P}(\bigcap_{i = 1}^{n} \{ S_{i} \geq 1 \})
         \\ &\mathbin{\phantom{=}} + \mathbb{E}[B \mid D > K - 1] \mathbb{P}(\bigcup_{i = 1}^{n} \{ S_{i} < 1 \})
         \\ &\geq \mathbb{E}[D \mid \bigcap_{i = 1}^{n} \{ S_{i} \geq 1 \}] \mathbb{P}(\bigcap_{i = 1}^{n} \{ S_{i} \geq 1 \})
         \\ &\mathbin{\phantom{=}} + \mathbb{E}[K - 1 \mid D > K - 1] \mathbb{P}(\bigcup_{i = 1}^{n} \{ S_{i} < 1 \})
         \\ &\geq \mathbb{E}[\min(D, K - 1) \mid \bigcap_{i = 1}^{n} \{ S_{i} \geq 1 \}] \mathbb{P}(\bigcap_{i = 1}^{n} \{ S_{i} \geq 1 \})
         \\ &\mathbin{\phantom{=}} + \mathbb{E}[\min(D, K - 1) \mid D > K - 1] \mathbb{P}(\bigcup_{i = 1}^{n} \{ S_{i} < 1 \})
         \\ &= \mathbb{E}[\min(D, K - 1)]
    \end{align*}
    and so $\mathbb{E}[B] / K \geq \frac{K - 1}{K} \mathbb{E}\left[\min\left( \frac{D}{K - 1}, 1 \right)\right]$. We thus have
    
    \begin{align*}
        \frac{\text{APX}}{\text{OPT}} \geq \min\left( \frac{K - 1}{K} \mathbb{E}\left[ \min\left( \frac{D}{K - 1}, 1 \right) \right], 1 - \delta \right) \text{.}
    \end{align*}
\end{proof}

For the purposes of bounding the welfare approximation ratio achieved by a posted-price mechanism in \Cref{welfare_bound}, we have the particular goal of finding a tail bound that is tractably invertible--that is, since the throughput $\throughput = \mathbb{E}\left[\min\left( \frac{D}{K - 1} , 1 \right) \right]$ shows up as a direct term in our bound on the approximation ratio, we need to find a lower bound on $\throughput$ in terms of the unavailability $\delta$ rather than an upper bound on unavailability $\delta$ in terms of throughput $\throughput$. This bound can be directly inserted into our lower bound on $\text{APX} / \text{OPT}$. For tractability purposes, this "inverse tail bound" should be writable in terms of elementary functions, while also remaining close to tight. As described in the previous section, to match the notation of \Cref{main_theorem}, which we will be using to construct these invertible tail bounds, we will define $\threshold = K - 1$.

One such bound can be achieved by selecting $f(x) = \exp(\lambda x) - 1$ and selecting the parameter $\lambda$ to make the bound as tight as possible. We present this bound in the following theorem:

\begin{theoremrep}\label{invertible_near_chernoff_bound}
    Let $D_{1}, ..., D_{n} \in [0, 1]$ be independent random variables with sum $D = \sum_{i = 1}^{n} D_{i}$ and let supply threshold $\kappa \geq 0$. For $\delta \triangleq \mathbb{P}\left( \bigcup_{i = 1}^{n} \{ D - D_{i} > \kappa \} \right)$, we then have the following bound on throughput $\throughput \triangleq \mathbb{E}[\min(D / \kappa, 1)]$:
    \begin{align*}
        \throughput \geq \frac{\frac{1}{\kappa} \ln\left( \exp\left( \kappa \sqrt{1 - \exp\left( \frac{1}{\kappa} \ln(\delta) \right)} \right) + (1 - \delta) \right)}{\exp\left( \sqrt{1 - \exp\left( \frac{1}{\kappa} \ln(\delta) \right)} - \frac{1}{\kappa} \ln(\delta) \right) - 1}
    \end{align*}
\end{theoremrep}

\begin{proof}
    The mapping $f(x) = \exp(\lambda x) - 1$ satisfies the requirements of \Cref{main_theorem}, so we have
    \begin{align*}
        \mathbb{P}(D \geq \kappa) \leq \frac{\mathbb{E}[\exp(\lambda Y) - 1]}{\exp(\lambda \kappa) - 1}
    \end{align*}
    for $Y \sim \text{Poisson}(\kappa \throughput)$, and since $\delta = \mathbb{P}\left( \bigcup_{i = 1}^{n} \{ D - D_{i} > \kappa \} \right) \leq \mathbb{P}(D \geq \kappa)$, we have
    \begin{align*}
        \delta \leq \frac{\mathbb{E}[\exp(\lambda Y) - 1]}{\exp(\lambda \kappa) - 1} \text{.}
    \end{align*}
    Multiplying both sides by $\exp(\lambda \kappa) - 1$ and adding $1$ to both sides, and we obtain
    \begin{align*}
        \delta \exp(\lambda \kappa) + (1 - \delta) \leq \mathbb{E}[\exp(\lambda Y)] \text{.}
    \end{align*}
    Using the formula for the moment-generating function of a Poisson-distributed variable, we have $\mathbb{E}[\exp(\lambda Y)] = \exp(\kappa \throughput (\exp(\lambda) - 1))$. Inserting this into the above inequality and rearranging gives us
    \begin{align*}
        \frac{\frac{1}{\kappa} \ln(\delta \exp(\lambda \kappa) + (1 - \delta))}{\exp(\lambda) - 1} \leq \throughput \text{.}
    \end{align*}
    This bound holds for any positive $\lambda$, but would be stronger for some selections of $\lambda$ than others. The optimal $\lambda$ does not have a closed form, but it can be very closely approximated. Consider that
    \begin{align*}
        \frac{\frac{1}{\kappa} \ln(\delta \exp(\lambda \kappa) + (1 - \delta))}{\exp(\lambda) - 1} \geq \frac{\frac{1}{\kappa} \ln(\delta \exp(\lambda \kappa))}{\exp(\lambda) - 1} \text{.}
    \end{align*}
    Close approximations to the optimal $\lambda$ on the right-hand side will tend to be underestimates of the optimal $\lambda$ on the left-hand side. Thus, we will demonstrate a bound on the optimal $\lambda$ for the left-hand side. We begin by taking the first derivative with respect to $\lambda$ and setting it equal to 0:
    \begin{align*}
        0 = \frac{d\left( \frac{\frac{1}{\kappa} \ln(\delta \exp(\lambda \kappa))}{\exp(\lambda) - 1} \right)}{d \lambda} = \frac{1}{\exp(\lambda) - 1} - \exp(\lambda) \frac{\frac{1}{\kappa} \ln(\delta \exp(\lambda \kappa))}{(\exp(\lambda) - 1)^{2}}
    \end{align*}
    Rearranging gives us
    \begin{align*}
        1 - \lambda - \exp(-\lambda) = \frac{1}{\kappa} \ln(\delta) \text{.}
    \end{align*}
    The mapping $\lambda \mapsto 1 - \lambda - \exp(-\lambda)$ is not invertible, but there is a known lower bound on the inverse mapping \parencite{RS-20} of
    \begin{align*}
        x \mapsto \sqrt{1 - \exp(x)} - x \text{.}
    \end{align*}
    Plugging $\frac{1}{\kappa} \ln(\delta)$ into this bound gives us an approximately optimal $\lambda$ of
    \begin{align*}
        \lambda = \sqrt{1 - \exp\left( \frac{1}{\kappa} \ln(\delta) \right)} - \frac{1}{\kappa} \ln(\delta)
    \end{align*}
    and inserting this selected $\lambda$ into our tail bound produces
    \begin{align*}
        \frac{\frac{1}{\kappa} \ln(\delta \exp(\left( \sqrt{1 - \exp\left( \frac{1}{\kappa} \ln(\delta) \right)} - \frac{1}{\kappa} \ln(\delta) \right) \kappa) + (1 - \delta))}{\exp\left( \sqrt{1 - \exp\left( \frac{1}{\kappa} \ln(\delta) \right)} - \frac{1}{\kappa} \ln(\delta) \right) - 1} \leq \throughput \text{,}
    \end{align*}
    which simplifies to
    \begin{align*}
        \frac{\frac{1}{\kappa} \ln\left( \exp\left( \kappa \sqrt{1 - \exp\left( \frac{1}{\kappa} \ln(\delta) \right)} \right) + (1 - \delta) \right)}{\exp\left( \sqrt{1 - \exp\left( \frac{1}{\kappa} \ln(\delta) \right)} - \frac{1}{\kappa} \ln(\delta) \right) - 1} \leq \throughput \text{.}
    \end{align*}
\end{proof}

The proof of \Cref{invertible_near_chernoff_bound} in the appendix is similar to the construction of a conventional Chernoff bound, though with some key differences.

We can insert the bound in \Cref{invertible_near_chernoff_bound} directly into our inequality on $\frac{\text{APX}}{\text{OPT}}$ from \Cref{welfare_bound}. Note that \Cref{main_theorem} also permits us to explore superior bounds on $\throughput$ using ReLU functions that, though less tractable, are significantly tighter. Additionally, we can use \Cref{main_theorem} with the same proof techniques as the conventional Chernoff bound. To understand the looseness of our tractable bound in \Cref{invertible_near_chernoff_bound}, we compare it to these other two bounds in \Cref{fig:chernoff-comparison} and \Cref{fig:chernoff-comparison-zoom}.
    
\begin{figure}[h]
    \center
    \includegraphics[scale=0.7]{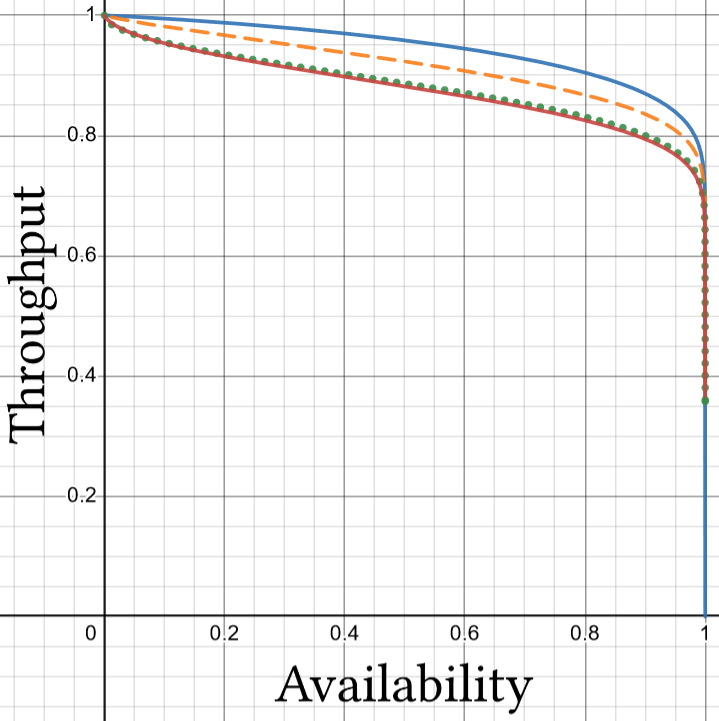}
    \caption{Graph of minimum throughput $\throughput = \mathbb{E}[\min(D / \threshold, 1)]$ as a function of availability $\availability = \mathbb{P}(D < \threshold)$ when $\threshold = 100$. Dashed orange line is produced when $f$ is a near-optimal ReLU function, the red line is produced when $f(x) = \exp(\lambda x) - 1$ for a carefully chosen $\lambda$, and the dotted green line is the conventional Chernoff-style bound. Compare to solid blue line, which represents the true $(\availability, \throughput)$ pairs for a Poisson distribution.}
    \label{fig:chernoff-comparison}
\end{figure}

\begin{figure}[h]
    \center
    \includegraphics[scale=0.6]{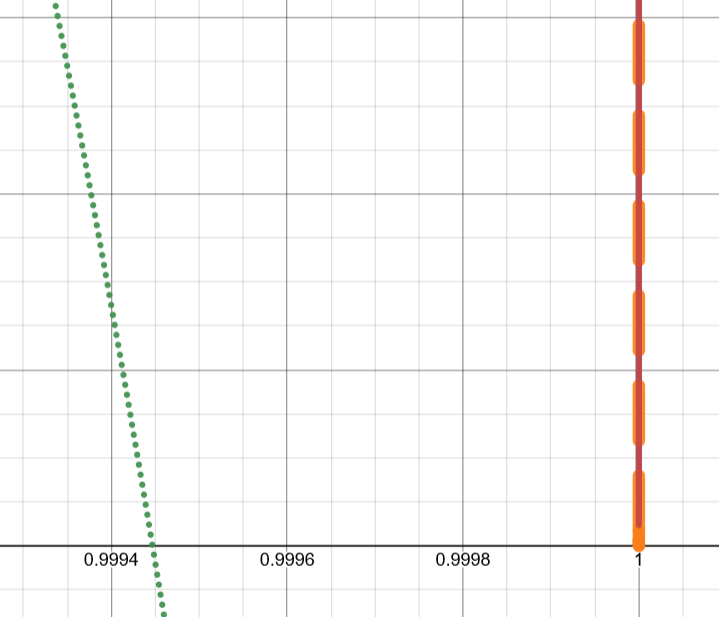}
    \caption{Zoomed in version of \Cref{fig:chernoff-comparison} on the point (availability, throughput) = (1, 0), with $\threshold = 5$. Observe that all bounds save for the dotted green conventional Chernoff-style bound very steeply approach (1, 0) as $\mathbb{P}(D < \threshold) \to 1$, whereas the Chernoff-style bound becomes negative.}
    \label{fig:chernoff-comparison-zoom}
\end{figure}

The first of our plotted bounds is the one from \Cref{invertible_near_chernoff_bound}, which is displayed as the red line in \Cref{fig:chernoff-comparison}. We begin by comparing it to a second bound, which is constructed using steps identical to those used to construct the traditional Chernoff bound; we select $f(x) = \exp(\lambda x)$, carefully select a value of $\lambda > 0$ equal to
\begin{align*}
    \lambda = \ln(1 / \throughput)
\end{align*}
to make the bound as strong as possible, and then weaken the result via a Pad\'{e} approximation. The result is a bound with the same functional form as the traditional Chernoff, but with the mean of the distribution $\mathbb{E}[D]$ replaced by the smaller quantity $\threshold \throughput \leq \mathbb{E}[D]$, where $\threshold \throughput$ is the absolute throughput:
\begin{align*}
    1 - \availability = \mathbb{P}(D \geq \kappa) \leq \exp\left( - \frac{\frac{1}{2} \left( \kappa - (\kappa \throughput) \right)^{2}}{(\kappa \throughput) + \frac{1}{3} \left( \kappa - (\kappa \throughput) \right)} \right)
\end{align*}
Inverting this bound produces
\begin{align*}
    \throughput \geq 1 + \frac{2}{3} \frac{1}{\kappa} \ln\left( \frac{1}{1 - \availability} \right) - \sqrt{\left( \frac{2}{3} \frac{1}{\kappa} \ln\left( \frac{1}{1 - \availability} \right) \right)^{2} + 2 \frac{1}{\kappa} \ln\left( \frac{1}{1 - \availability} \right)} \text{,}
\end{align*}
which we plot as the dotted green line in \Cref{fig:chernoff-comparison}.

Our first bound has a major advantage over the traditional Chernoff-style bound: Observe how in \Cref{fig:chernoff-comparison-zoom}, which is very zoomed in, that the traditional Chernoff-style bound produces a negative lower bound on throughput $\throughput$ for very high availability $\mathbb{P}(D < \threshold) \approx 1$. While Chernoff bounds are very effective in the asymptotic case where the threshold $\threshold \to \infty$, they fall short at bounding very low tail probabilities for fixed $\threshold$. A failure to bound low tail probabilities is an especially big issue when $\threshold$ is very low, where the range of tail probabilities with a negative lower bound on the throughput $\throughput$ can be quite large. It is also a problem for analyzing posted-price mechanisms designed to have very low probabilities of running out of supplies, since it means the traditional Chernoff bound provides no welfare guarantees as availability $\mathbb{P}(D < \kappa) \to 1$.

We compare both of these bounds to a third bound, which is constructed with an optimal ReLU function $x \mapsto \max(x - \reluparam, 0)$ that we can use with \Cref{main_theorem}. The value of $\reluparam$ is adaptively selected, based on the desired value of $\mathbb{P}(D \geq \threshold)$, to make the bound
\begin{align*}
    \mathbb{P}(D \geq \kappa) \leq \frac{\mathbb{E}[\max(Y - \reluparam, 0)]}{\max(\kappa - \reluparam, 0)}
\end{align*}
as strong as possible for $Y \sim \text{Poisson}(\threshold \throughput)$. We display the bound thus obtained as a dashed orange line in \Cref{fig:chernoff-comparison}, to showcase how much lower the welfare guarantees of a posted-price mechanism become when using a non-ReLU function $f$.

Lastly, all three of these bounds are compared to the true $\text{availability-throughput}$ pairs of Poisson distributions. This is the solid blue line in \Cref{fig:chernoff-comparison}. Since the three bounds are lower bounds on throughput $\throughput$ in terms of availability $\availability$, it is clear that the actual throughput $\throughput$ of a Poisson distribution for any fixed $\availability$ should be higher than these lower bounds. Thus, this solid blue line should be seen as a benchmark to compare against our three lower bounds on the throughput $\throughput$; the smaller the gap between the true throughput $\throughput$ of a Poisson distribution and a given lower bound on $\throughput$, the better. Our goal is to minimize this gap as much as possible while keeping our bounds analytically tractable. As one can see, our optimal ReLU bound (the dashed orange line) minimizes this gap the most, but is quite intractable. The other two bounds perform worse, but are far more tractable. Thus, there is a clear tradeoff between tractability and strength of our three bounds.

\begin{figure}
    \center
    \begin{subfigure}[h]{0.4\linewidth}
        \includegraphics[width=\linewidth]{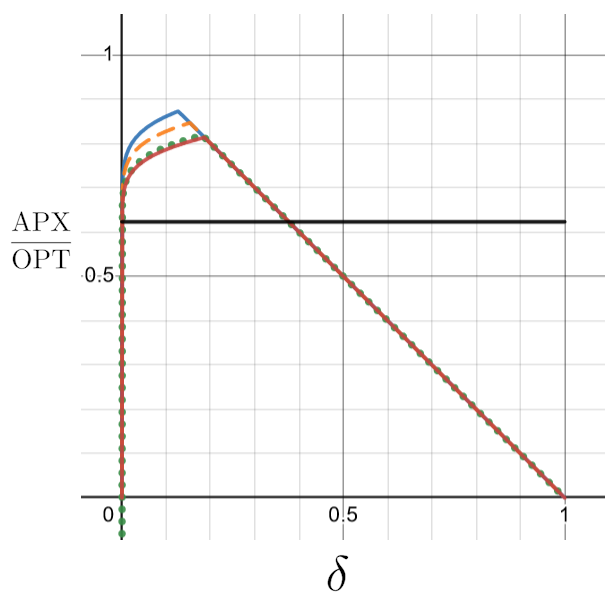}
        \caption{Welfare bounds given the same parameters as \Cref{fig:chernoff-comparison}, which has supply $K = 101$.}
        \label{fig:welfare-big-vs-prev}
    \end{subfigure}
    \hfill
    \begin{subfigure}[h]{0.4\linewidth}
        \includegraphics[width=\linewidth]{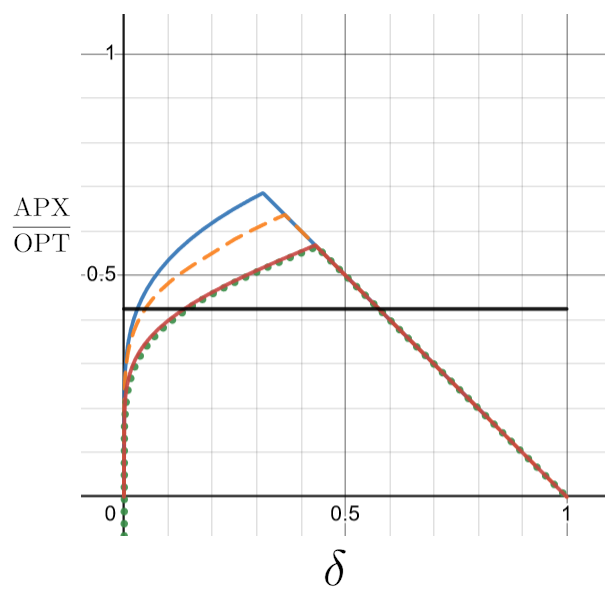}
        \caption{Welfare bounds in a low-supply setting, with supply $K = 10$.}
        \label{fig:welfare-small-vs-prev}
    \end{subfigure}%
    \caption{Plugging the same bounds on $\throughput$ as in \Cref{fig:chernoff-comparison} into our lower bound on $\frac{\text{APX}}{\text{OPT}}$. The black horizontal line is the $\Big(1 + \sqrt{\frac{8 \ln(K)}{K}}\Big)^{-1}$ lower bound on $\frac{\text{APX}}{\text{OPT}}$ obtained by \textcite{HKS-07}. Note that this welfare bound is achieved by a particular $\delta$ value, and we have drawn it as a horizontal black line to better illustrate for which values of $\delta$ our own bounds produce stronger or weaker welfare guarantees.}
    \label{fig:welfare-comparison-vs-prev}
\end{figure}

Recall that by defining our supply threshold $\threshold = K - 1$, each of the bounds on throughput $\throughput = \mathbb{E}[\min(D / \threshold, 1)]$ (and the curve corresponding to the Poisson distribution) can be inserted into our welfare bound
\begin{align*}
    \frac{\text{APX}}{\text{OPT}} &\geq \min\left( \frac{K - 1}{K} \mathbb{E}\left[ \min\left( \frac{D}{K - 1}, 1 \right) \right], 1 - \delta \right)
    \\ &= \min\left( \frac{K - 1}{K} \throughput, 1 - \delta \right)
\end{align*}
from \Cref{welfare_bound} to obtain the tradeoff curves displayed in \Cref{fig:welfare-comparison-vs-prev}. Any particular unavailability $\delta$ produces a different lower bound on $\frac{\text{APX}}{\text{OPT}}$, and the solid blue Poisson curve is an upper bound on any such approximation ratio. Thus, the true tradeoff curve must lie between the strongest of our lower bounds and the Poisson curve. We plot these tradeoffs for both a high supply environment (\Cref{fig:welfare-big-vs-prev}) and a low supply environment (\Cref{fig:welfare-small-vs-prev}) to give a visualization of the changes that occur as the parameter $K$ varies.

As can be seen, there is a minor plateau around the peak of the welfare curves $\frac{K - 1}{K} \throughput \approx 1 - \delta$, where there is a minimal penalty to the lower bound on $\frac{\text{APX}}{\text{OPT}}$ for decreasing unavailability $\delta$, followed by a sharp dropoff that occurs as the unavailability $\delta$ moves far below the welfare-optimal level. The slope of the plateau depends on which of the bounds is used, increasing as the bounds become stronger and move toward the Poisson curve. Furthermore, the slope of the plateau becomes smaller as we move toward high-supply environments. This suggests that a mechanism designer who is willing to sacrifice some of the welfare guarantee of a posted-price mechanism in exchange for lowering unavailability $\delta$ can do so at a very favorable tradeoff rate.

There are also implications for mechanism designers unconcerned with minimizing unavailability $\delta$, and only concerned with maximizing the worst-case welfare guarantee. The black line in \Cref{fig:welfare-big-vs-prev} and \Cref{fig:welfare-small-vs-prev} is the $\frac{\text{OPT}}{\text{APX}} \leq 1 + \sqrt{\frac{8 \ln(K)}{K}}$ bound obtained by \textcite{HKS-07}, which is intended to be used in a unit-demand setting with asymptotically high values of the real supply $K$. There is a gap between the best lower bounds obtained on $\frac{\text{APX}}{\text{OPT}}$ via \Cref{main_theorem} at the welfare-optimal selection of unavailability $\delta$ and this black line, which suggests that the constants involved for the $1 + \sqrt{\frac{8 \ln(K)}{K}}$ bound can be significantly improved at these lower values of supply $K$. Therefore, posted-price mechanisms can provide much stronger welfare guarantees at non-asymptotically large $K$ than the bounds from \textcite{HKS-07} permit.

\section{Transaction fee mechanism design}

We take a closer look at the application of our results to designing transaction fee mechanisms, which are used in blockchain technology to allocate limited space on a block.

As discussed in the introduction, space in a block is considered a valuable resource since each block can accommodate only so many units of data (or gas, in the case of Ethereum). A transaction fee mechanism is deployed to decide the inclusion of transactions in a block. Ethereum, for instance, runs a posted-price mechanism that burns all the payments collected from users \parencite{TR-20}. Conditioned on demand at the price being smaller than capacity of the block, the posted-price mechanism is strategy-proof for users and block proposers, while also deterring collusion between users and block proposers. However, when demand exceeds capacity, an emergency mechanism is deployed wherein the three incentive properties satisfied by the posted-price mechanism cannot all be satisfied simultaneously. The first price auction used by Ethereum as the emergency mechanism is not strategy-proof for users.

Having a high availability is desirable since it corresponds to not needing to frequently deploy the emergency mechanism. However, high availability trades off against throughput. With Bitcoin, where block capacity is an exogenous constraint and absolute throughput is a decision variable, high throughput corresponds to an efficient utilization of the valuable block space. Meanwhile with Ethereum, where absolute throughput is an exogenous constraint and block capacity is a decision variable, high throughput corresponds to low block capacity, and thereby, a smaller amount of information to process for block proposers in periods of high demand.

An ideal price-adjustment mechanism would like to directly control for availability. However, measuring availability is not straightforward since whether the demand exceeds the capacity of a block is binary; either the demand exceeds capacity, or it does not. Instead, price-posting heuristics use throughput and absolute throughput as an easily measurable lever to indirectly ensure high availability. Ethereum targets maintaining an average absolute throughput of 15 million gas, and achieves this through a heuristic that increases (decreases) the price if the ex-post absolute throughput of the previous block was more (less) than the 15 million gas target. Ethereum supports up to 30 million gas in each block, and thus, the heuristic translates to maintaining a 50\% throughput. Improved methods for managing this tradeoff have been proposed \parencite{FMPS-21}.

Our analysis of the shape of the tradeoff curve suggests a scope to reduce the block capacity (and thus, increase the throughput) without a significant decrease in availability\footnote{In practice, the likely effects on availability are less clear. In Ethereum, the emergency auction is triggered in around 2.3\% of blocks \parencite{CRS24, LLNZZZ-22}, corresponding to 97.7\% availability. This is lower than what our tradeoff curves would permit if throughput were fixed to $1/2$, and suggests that Ethereum's price adjustment to maintain its $1/2$ throughput target is not fast enough. When the price lags behind its ideal level, blocks may achieve a throughput either far above or below the $1/2$ target, and it is the average availability between these two cases which produces the 97.7\% figure.}. The target absolute throughput is 15 million gas, while 750,000 is a conservative upper bound on the maximum gas requested by any single transaction. Thus, the effective supply is at least $40$. Targeting a 50\% throughput would guarantee an availability of at least 99.9\%. On the other hand, for a reduced supply of 26.25 million gas (corresponding to an effective supply of 35), a target absolute throughput of 15 million gas would correspond to a 57\% throughput. Our bounds on availability, given this smaller effective supply and larger throughput target, would guarantee an availability of at least 99.7\%, which is only a minor decrease from the status quo.


For the parameters mentioned in the previous paragraph, we can also examine the case where block capacity is exogenously fixed and the absolute throughput is a decision variable (like in Bitcoin). \Cref{fig:ethereum-tradeoff} plots the availability vs throughput tradeoff for $\threshold = 40$, using an optimal ReLU bound from \Cref{main_theorem}. It is possible to select different points along the tradeoff curve using different throughput targets. For minor sacrifice in availability, a significant gain in throughput can be realized. For instance, a throughput of $60\%$ guarantees an availability of at least $99.66\%$, which translates to requiring the emergency auction approximately once every $300$ blocks. A throughput of $60\%$ would also guarantee approximately $60\%$ of the optimal welfare (\Cref{welfare_bound}). As with Ethereum, this is a only minor decrease in availability from the status quo.

\begin{figure}[h]
    \center
    \includegraphics[scale=0.6]{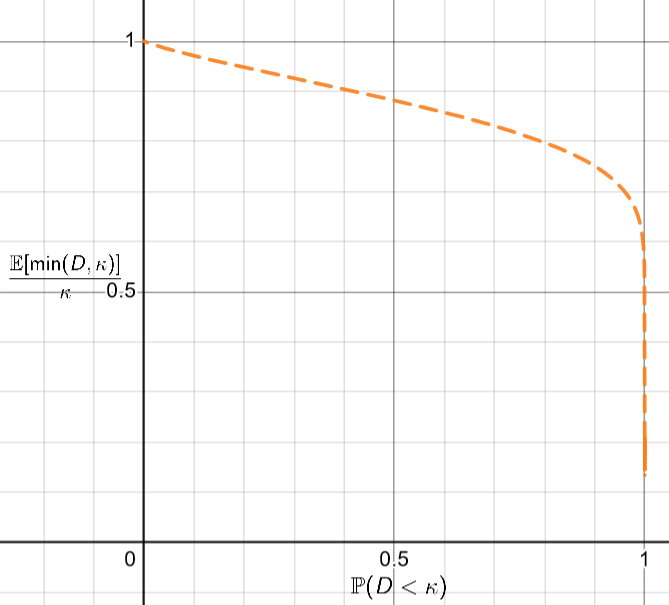}
    \caption{The optimal ReLU bound from \Cref{main_theorem} when $\threshold = 40$, illustrating the tradeoff curve between guarantees on throughput and availability for Ethereum. Throughput can be set to $60\%$ while maintaining a $99.66\%$ availability.}
    \label{fig:ethereum-tradeoff}
\end{figure}

\section{Conclusions and future work}
For posted-price mechanisms, we are able to significantly improve upon the lower bounds on expected welfare provided by a conventional Chernoff bound. Our bounds on throughput in terms of availability also improve on the analytical tractability of those provided by \textcite{CDL-23} without significantly sacrificing strength, and can be applied to settings with up-to-unit-demand agents.

We are able to show that the unavailability $\delta$ can be made much smaller than its usual value in prophet inequalities, with only a minor penalty to expected welfare. As can be seen from \Cref{fig:chernoff-comparison} and \Cref{fig:welfare-comparison-vs-prev}, several of our derived lower bounds on throughput $\throughput$ are nearly flat as a function of $\delta$ near the peak of the welfare curve. A nearly flat curve suggests that prices intentionally selected to be higher than their welfare-optimal value can significantly decrease unavailability $\delta$ with only a minor impact on throughput $\throughput$, and thus only a minor impact on expected welfare. A mechanism designer has good reason to make $\delta$ as small as possible, and so may wish to take advantage of this favorable tradeoff.

The bounds we derive can almost certainly be made tighter, and this is an important area for future work. \textcite{CDL-23} are able to demonstrate that Poisson-distributed demand produces the worst ratio of expected welfare achieved by a posted-price mechanism compared to a prophet in a unit-demand setting. In comparison, as can be seen from \Cref{fig:chernoff-comparison}, the throughput $\throughput$ corresponding to a true Poisson distribution for any given availability $\availability$ is above the strongest lower bounds on $\throughput$ available from \Cref{main_theorem}, via selecting an optimal ReLU function. The tradeoff curve  for Poisson distributions would need to be identical to the strongest bound from \Cref{main_theorem} for the proof techniques demonstrated in this paper to extend the results of \textcite{CDL-23} that Poisson demand is worst-case from the unit-demand setting to the up-to-unit-demand setting. Future work should seek to close the gap between the two tradeoff curves.

Another area for future work is to refine our availability-throughput bounds to better capture salient features of the demand distribution in Ethereum blocks. The concentration inequalities presented in this paper assume only that individual demands are independent and up-to-unit, and if stronger assumptions can be made on the individual demands' distributions, stronger concentration inequalities could be developed. For example, in Ethereum it is known that the gas requested by many transactions falls below their maximum possible value. Our analysis used a conservative upper bound on the gas requested by any single transaction of $750,000$. However, a standard transfer between two Ethereum wallets requires $21,000$ units of gas, which is an order of magnitude smaller. Since the variability in the gas requested by each transaction is high, the "effective supply" of $\threshold = 40$ that we selected is an underestimate; most blocks typically accommodate over $100$ transactions. A more refined analysis that takes into account these empirical facts could demonstrate an improved tradeoff between throughput and availability, by restricting attention to more realistic demand distributions rather than the worst-case demand distributions examined in this paper.

\printbibliography

\end{document}